\documentclass[orivec]{llncs}

\usepackage[margin=1in]{geometry}

\usepackage{proof}
\usepackage{latexsym}
\usepackage{amssymb}
\usepackage{amsmath}
\usepackage{url}
\usepackage{cmll}
\usepackage{MnSymbol}
\usepackage{graphicx}
\usepackage{bussproofs}
\usepackage[all]{xy}
\usepackage{comment}
\usepackage{lscape}
\usepackage{ntheorem}

\renewtheorem{theorem}{Theorem}
\renewtheorem{lemma}[theorem]{Lemma}
\renewtheorem{corollary}[theorem]{Corollary}
\renewtheorem{proposition}[theorem]{Proposition}
\renewtheorem{definition}[theorem]{Definition}

\def\punit{\bot}
\def\tunit{\mathrm{I}}
\def\sunit{\cdot}
\def\limp{\multimap}
\def\excl{\leftY}
\def\ltens{\otimes}
\def\merge{\bullet}

\def\seq{\Rightarrow}

\def\fill{\mathrm{FILL}}
\def\biill{\mathrm{BiILL}}
\def\biilldc{\mathrm{BiILL}{\scriptstyle dc}}
\def\biillsn{\mathrm{BiILL}{\scriptstyle sn}}
\def\biilldn{\mathrm{BiILL}{\scriptstyle dn}}
\def\filldc{\mathrm{FILL}{\scriptstyle dc}}
\def\filldn{\mathrm{FILL}{\scriptstyle dn}}

\def\drp{\mathit{drp}}
\def\rp{\mathit{rp}}

\newcommand{\binop}{\heartsuit}

\newcommand \seqtodisp[1]{\ulcorner{#1}\urcorner}
\newcommand\disptoseq[1]{\llcorner{#1}\lrcorner}

\def\Sc{\mathcal{S}}
\def\Tc{\mathcal{T}}
\def\Uc{\mathcal{U}}
\def\Vc{\mathcal{V}}
\def\Xc{\mathcal{X}}
\def\Yc{\mathcal{Y}}

\def\turn{\vdash}

\def\adcn{\wedge}

\def\adtp{\top}

\def\mlim{\multimap}

\def\sequent{\vdash}

\title{Annotation-Free Sequent Calculi for Full Intuitionistic Linear Logic
-- Extended Version}
\author{Ranald Clouston \and Jeremy Dawson \and Rajeev Gor\'e 
\and Alwen Tiu
}
\institute{Logic and Computation Group, Research School of Computer Science,\\
The Australian National University, Canberra ACT 0200, Australia}

\begin{document}

\maketitle

\begin{abstract}
Full Intuitionistic Linear Logic (FILL) is multiplicative
intuitionistic linear logic extended with par. Its proof theory has
been notoriously difficult to get right, and existing sequent calculi
all involve inference rules with complex annotations to guarantee
soundness and cut-elimination. We give a simple and annotation-free
display calculus for FILL which satisfies Belnap's generic
cut-elimination theorem. To do so, our display calculus actually
handles an extension of FILL, called Bi-Intuitionistic Linear Logic
(BiILL), with an `exclusion' connective defined via an adjunction with
par. We refine our display calculus for BiILL into a cut-free nested
sequent calculus with deep inference in which the explicit structural
rules of the display calculus become admissible. A separation property
guarantees that proofs of FILL formulae in the deep inference calculus
contain no trace of exclusion.  Each such rule is sound for the
semantics of FILL, thus our deep inference calculus and display
calculus are conservative over FILL. The deep inference calculus also
enjoys the subformula property and terminating backward proof search,
which gives the NP-completeness of BiILL and FILL.
\end{abstract}

\section{Introduction}

Multiplicative Intuitionistic Linear Logic (MILL) contains as connectives only
tensor $\otimes$, its unit $I$, and its residual $\limp$,
where we use $I$ rather than the usual $1$ to avoid a clash with the
categorical notation for terminal object.
The connective par $\parr$ and its unit $\bot$ are traditionally
only introduced when we move to classical Multiplicative Linear Logic (MLL),
but Hyland and de Paiva's Full Intuitionistic
Linear Logic (FILL) \cite{Hyland:Full} shows that a sensible notion of
par can be added to MILL without collapse to classicality. FILL's semantics are
categorical, with the interaction between the $(\otimes,I,\limp)$
and $(\parr,\bot)$ fragments entirely described by the equivalent formulae
shown below:
\begin{equation}\label{eq:wdis}
(p\otimes (q\parr r)) \,\limp\, ((p\otimes q)\parr r)
\qquad\qquad\qquad
((p\limp q)\parr r) \,\limp\, (p\limp(q\parr r))
\end{equation}
The first formula is variously called \emph{weak distributivity}
\cite{Hyland:Full,Cockett:Weakly}, \emph{linear distributivity}
\cite{Cockett:Proof}, and \emph{dissociativity} \cite{Dosen:Proof}. The second
we call Grishin~(b)~\cite{Gore:SLD}. Its converse, called Grishin
(a), is not FILL-valid, and indeed adding it to FILL recovers MLL.

From a traditional sequent calculus perspective, FILL is the logic specified
by taking a two-sided sequent calculus for MLL, which enjoys
cut-elimination, and restricting its $(\limp R_2)$
rule to apply only to ``singletons on the right'', giving $(\limp
R_1)$, as shown below:
\[
 \AxiomC{$\Gamma, A \sequent B$}
 \LeftLabel{$(\mlim R_1)$}
 \UnaryInfC{$\Gamma \sequent A \mlim B$}
 \DisplayProof
\qquad\qquad\qquad\qquad
 \AxiomC{$\Gamma, A \sequent B, \Delta$}
 \LeftLabel{$(\mlim R_2)$}
 \UnaryInfC{$\Gamma \sequent A \mlim B, \Delta$}
 \DisplayProof
\]
Since exactly this restriction converts Gentzen's LK for ordinary
classical logic to Gentzen's LJ for intuitionistic logic,
FILL arises very naturally. Unfortunately
the resulting calculus fails cut-elimination~\cite{Schellinx:Some}%
.
(Note that there is also work on natural deduction and proof nets for FILL
\cite{Cockett:Proof,Bellin:Subnets,Martini:Experiments,dePaiva:Parigot}. In
this setting the problems of cut-elimination are side-stepped; see
the discussion of ``essential cuts'' in \cite{Cockett:Proof} in
particular.)

Hyland and de Paiva~\cite{Hyland:Full} therefore sought a middle
ground between the too weak $(\limp R_1)$ and the unsound $(\limp
R_2)$ by annotating formulae with \emph{term assignments}, and using
them to restrict the application of $(\limp R_2)$ - the restriction
requires that the variable typed by $A$ not appear free in the
terms typed by $\Delta$. Reasoning with freeness in the presence of
variable binders is notoriously tricky, and a bug was subsequently
found by Bierman \cite{Bierman:Note} which meant that the proof of the
sequent below requires a cut that is not eliminable:
\begin{equation}\label{eq:Bierman}
  (a\parr b)\parr c \turn a,(b\parr c\limp d)\parr e\limp d\parr e
\end{equation}
Bierman~\cite{Bierman:Note} presented two possible corrections to
the term assignment system, one due to Bellin. These were subsequently
refined by Br\"{a}uner and de Paiva~\cite{Brauner:Formulation} to replace
the term assignments by rules annotated with a binary relation between
formulae on the left and on the right of the turnstile, which effectively
trace variable occurrence. The only existing annotation-free sequent
calculi for FILL~\cite{DBLP:conf/obpdc/GalmicheB95,Gore:SLD} are
incorrect. The first~\cite{DBLP:conf/obpdc/GalmicheB95} uses $(\limp
R_2)$ without the required annotations, making it unsound, and
also contains other transcription errors. The second~\cite{Gore:SLD}
identifies FILL with `Bi-Linear Logic', which fails weak distributivity
and has an extra connective called `exclusion', of which more shortly.

The existing correct annotated sequent
calculi~\cite{Bierman:Note,Brauner:Formulation} have some
weaknesses. First, the introduction rules for a connective do not
define that connective in isolation, as was Gentzen's ideal. Instead,
they introduce $\limp$ on the right only when the context in which the rule
sits obeys the rule's side-condition. A consequence is that they
cannot be used for naive backward proof search since we must apply the
rule upwards blindly, and then check the side-conditions once we have
a putative derivation. Second, the term-calculus that results from the
annotations has not been shown to have any computational content since its sole
purpose is to block unsound inferences by tracking variable
occurrence~\cite{Brauner:Formulation}. Thus, $\fill$'s close
relationship with other logics is obscured by these complex
annotational devices, leading to it being described as
proof-theoretically ``curious''~\cite{Cockett:Proof}, and leading
others to conclude that $\fill$ ``does not have a satisfactory proof
theory''~\cite{Chaudhuri:Inverse}. 

We believe these difficulties 
arise because efforts have focused on an `unbalanced'
logic. We show that adding an `exclusion'
connective $\excl$, dual to
$\limp$, gives a fully `balanced' logic, which we call
Bi-Intuitionistic Linear Logic (BiILL). The beauty of $\biill$ is that
it has a simple {\em display calculus}~\cite{Belnap82JPL,Gore:SLD} $\biilldc$
that inherits Belnap's general cut-elimination theorem ``for free''.
A similar situation has already been observed in classical modal
logic, where it has proved impossible to extend traditional Gentzen
sequents to a uniform and general proof-theory encompassing the
numerous extensions of normal modal logic K. Display calculi capture a
large class of such modal extensions uniformly and
modularly~\cite{DBLP:journals/logcom/Wansing94,kracht-power} by
viewing them as fragments of (the display calculi for) tense logics,
which conservatively extend modal logic with two modalities
$\blacklozenge$ and $\blacksquare$,
respectively adjoint to
the original $\Box$ and $\Diamond$.

In tense logics, the conservativity result is trivial since both modal
and tense logics are defined with respect to the same Kripke
semantics. 
With BiILL and FILL,
however, there is no such existing conservativity result via semantics.
The conservativity of BiILL over FILL would follow 
if we could show that a derivation of a FILL formula in $\biilldc$
preserved FILL-validity downwards: unfortunately, this does not
hold, as explained next.

Belnap's generic cut-elimination procedure applies to $\biilldc$
because of the ``display property'', whereby any substructure of a
sequent can be displayed as the whole of either the antecedent or
succedent. The display property for $\biilldc$ is obtained via certain
reversible structural rules, called {\em display rules}, which encode
the various adjunctions between the connectives, such as the one
between par and exclusion. Any $\biilldc$-derivation of a FILL formula
that uses this adjunction to display a substructure contains
occurrences of a structural connective which is an exact proxy
for exclusion. That is, a $\biilldc$-derivation of a FILL formula may
require inference steps that have no meaning in FILL, thus
we cannot use our display calculus $\biilldc$ directly to show
conservativity of BiILL over FILL. We circumvent this problem by
showing that the structural rules to maintain the display property become
{\em admissible}, provided one uses {\em deep inference}.

Following a methodology established for bi-intuitionistic and tense
logics~\cite{Gore10AiML,Gore11LMCS}, we show that the display calculus
for $\biill$ can be refined to a {\em nested sequent
  calculus}~\cite{kashima-cut-free-tense,Brunnler09AML}, called
$\biilldn$, which contains no explicit structural rules, and hence no
cut rule, as long as its introduction rules can act ``deeply'' on any
substructure in a given structure.  To prove that $\biilldn$ is sound
and complete for $\biill$, we use an intermediate nested sequent
calculus called $\biillsn$ which, similar to our display calculus, has
explicit structural rules, including cut, and uses {\em shallow}
inference rules that apply only to the topmost sequent in a
nested sequent. Our shallow inference calculus $\biillsn$
can simulate cut-free proofs of our display calculus $\biilldc$, and
vice versa. It enjoys cut-elimination, the display
property and coincides with the deep-inference calculus $\biilldn$
with respect to (cut-free) derivability. Together these imply that
$\biilldn$ is sound and (cut-free) complete for BiILL. Our deep nested
sequent calculus $\biilldn$ also enjoys a {\em separation property}: a
$\biilldn$-derivation of a formula $A$ uses only introduction rules
for the connectives appearing in $A.$ By selecting from $\biilldn$
only the introduction rules for the connectives in FILL, we obtain a
nested (cut-free and deep inference) calculus $\filldn$ which is
complete for FILL.  We then show that the rules of $\filldn$ are also
sound for the semantics of FILL. The conservativity
of BiILL over FILL follows since a FILL formula $A$ which is
valid in BiILL will be cut-free derivable in $\biilldc$, and hence in
$\biilldn$, and hence in $\filldn$, and hence valid in FILL.


Viewed upwards, introduction rules for display calculi use shallow
inference and can require disassembling structures into an appropriate
form using the display rules, meaning that display calculi do not
enjoy a ``substructure property''. The modularity of display calculi
also demands explicit structural rules for associativity,
commutativity and weak-distributivity. These necessary aspects of
display calculi make them unsuitable for proof search since the
various structural rules and reversible rules can be applied
indiscriminately. As structural rules are admissible in the nested
deep inference calculus $\biilldn$, proof search in it is easier to
manage than in the display calculus. Using $\biilldn$, we show that
the tautology problem for $\biill$ and $\fill$ are in fact
NP-complete.


\emph{We gratefully acknowledge the comments of the anonymous reviewers. 
This work is partly supported by the ARC Discovery Projects DP110103173
and DP120101244.}






\section{Display Calculi}\label{Sec:display}


\subsection{Syntax}

\begin{definition}
BiILL-\emph{formulae} are defined using the grammar below
where $p$ is from some fixed set of propositional variables 
$$
  A ::= p \mid I \mid \bot \mid A\otimes A \mid A\parr A \mid A\limp A \mid A\excl A
$$
\emph{Antecedent} and \emph{succedent} BiILL-\emph{structures} (also
known as antecedent and succedent parts) are defined by
mutual induction, where $\Phi$ is a structural constant and
$A$ is a BiILL-formula:
$$
  X_a ::= A \mid \Phi \mid X_a,X_a \mid X_a<X_s
\qquad
\qquad
\qquad
  X_s ::= A \mid \Phi \mid X_s,X_s \mid X_a>X_s
$$
FILL-formulae are BiILL-formulae with no occurrence of the exclusion
connective $\excl$. FILL-structures are $\biill$-structures with
no occurrence of $<$, and containing only FILL-formulae.
We stipulate that $\otimes$ and $\parr$ bind tighter than $\limp$ and
$\excl$, that comma binds tighter than $>$ and $<$, and resolve
$A\limp B\limp C$ as $A\limp(B\limp C)$.
A BiILL- (resp. FILL-) \emph{sequent} is a pair comprising an antecedent and a
succedent BiILL- (resp. FILL-) structure, written $X_a\turn X_s$. 
\end{definition}

\begin{definition}\label{def:tau}
We can translate sequents $X\turn Y$ into formulae as $\tau^a(X)\limp\tau^s(Y)$,
given the mutually inductively defined antecedent and succedent
$\tau$-\emph{translations}:
\[
\begin{tabular}{|c|c|c|c|c|c|}
\hline & $A$ & $\Phi$ & $X,Y$ & $X>Y$ & $X<Y$ \\
\hline $\tau^a$ & $A$ & $I$ & $\tau^a(X)\otimes\tau^a(Y)$ & & $\tau^a(X)\excl\tau^s(Y)$ \\
\hline $\tau^s$ & $A$ & $\bot$ & $\tau^s(X)\parr\tau^s(Y)$ & $\tau^a(X)\limp\tau^s(Y)$ & \ \\ \hline
\end{tabular}
\]
Hence $\Phi$ and comma are \emph{overloaded} to be translated into
different connectives depending on their position%
. By uniformly replacing our
structural connective $<$ with $>$, we could have also 
overloaded $>$ to stand for $\limp$ and $\excl$,
which would have avoided the blank spaces in the above table, but we have
opted to use different connectives to help visually emphasise whether
a given structure lives in BiILL or its fragment FILL.
\end{definition}

The display calculi for FILL and BiILL are given in Fig.~\ref{Fig:dc}.

\begin{figure}[t]
  \begin{tabular}[c]{l@{\extracolsep{3cm}}l}
\multicolumn{2}{l}{Cut and identity:}
  \\[0px]
  \qquad  \qquad
  \AxiomC{(id)\quad$p\turn p$}
  \DisplayProof
  &
  \AxiomC{$X\turn A$}
  \AxiomC{$A\turn Y$}
  \LeftLabel{(cut)\,}
  \BinaryInfC{$X\turn Y$}
  \DisplayProof
\\
\multicolumn{2}{l}{Logical rules:}
\\[0px]
\qquad\qquad
  \AxiomC{$\Phi\turn X$}
  \LeftLabel{($I\turn$)\,}
  \UnaryInfC{$I\turn X$}
  \DisplayProof
  &
  \AxiomC{($\turn I$)\quad$\Phi\turn I$}
  \DisplayProof
\\[10px]
  \qquad\qquad
  \AxiomC{($\bot\turn$)\quad$\bot\turn\Phi$}
  \DisplayProof
  &
  \AxiomC{$X\turn\Phi$}
  \LeftLabel{($\turn\bot$)\,}
  \UnaryInfC{$X\turn\bot$}
  \DisplayProof
\\[10px]
  \qquad\qquad
  \AxiomC{$A,B\turn X$}
  \LeftLabel{($\otimes\turn$)\,}
  \UnaryInfC{$A\otimes B\turn X$}
  \DisplayProof
  &
  \AxiomC{$X\turn A$}
  \AxiomC{$Y\turn B$}
  \LeftLabel{($\turn\otimes$)\,}
  \BinaryInfC{$X,Y\turn A\otimes B$}
  \DisplayProof
\\[10px]
  \qquad\qquad
  \AxiomC{$A\turn X$}
  \AxiomC{$B\turn Y$}
  \LeftLabel{($\parr\turn$)\,}
  \BinaryInfC{$A\parr B\turn X,Y$}
  \DisplayProof
  &
  \AxiomC{$X\turn A,B$}
  \LeftLabel{($\turn\parr$)\,}
  \UnaryInfC{$X\turn A\parr B$}
  \DisplayProof
\\[10px]
  \qquad\qquad
  \AxiomC{$X\turn A$}
  \AxiomC{$B\turn Y$}
  \LeftLabel{($\limp\turn$)\,}
  \BinaryInfC{$A\limp B\turn X>Y$}
  \DisplayProof
  &
  \AxiomC{$X\turn A>B$}
  \LeftLabel{($\turn\limp$)\,}
  \UnaryInfC{$X\turn A\limp B$}
  \DisplayProof
\\
\multicolumn{2}{l}{Structural rules:}
\\[5px]
\multicolumn{2}{c}{
  \AxiomC{$X\turn Y>Z$}
  \doubleLine
  \LeftLabel{(rp)\,}
  \UnaryInfC{$X,Y\turn Z$}
  \DisplayProof
  \qquad
  \AxiomC{$X, Y \turn Z$}
  \doubleLine
  \LeftLabel{(rp)\,}
  \UnaryInfC{$Y\turn X>Z$}
  \DisplayProof
  \qquad
  \AxiomC{$X<Y\turn Z$}
  \doubleLine
  \LeftLabel{(drp)\,}
  \UnaryInfC{$X\turn Y,Z$}
  \DisplayProof
  \qquad
  \AxiomC{$X \turn Y,Z$}
  \doubleLine
  \LeftLabel{(drp)\,}
  \UnaryInfC{$X<Z\turn Y$}
  \DisplayProof
}
\\[10px]
  \qquad\qquad
  \AxiomC{$X,\Phi\turn Y$}
  \doubleLine
  \LeftLabel{($\Phi\turn$)\,}
  \UnaryInfC{$X\turn Y$}
  \DisplayProof
  &
  \AxiomC{$X\turn \Phi,Y$}
  \doubleLine
  \LeftLabel{($\turn\Phi$)\,}
  \UnaryInfC{$X\turn Y$}
  \DisplayProof
\\[10px]
  \qquad\qquad
  \AxiomC{$W,(X,Y)\turn Z$}
  \doubleLine
  \LeftLabel{(Ass $\turn$)\,}
  \UnaryInfC{$(W,X),Y\turn Z$}
  \DisplayProof
  &
  \AxiomC{$W\turn (X,Y),Z$}
  \doubleLine
  \LeftLabel{($\turn$ Ass)\,}
  \UnaryInfC{$W\turn X,(Y,Z)$}
  \DisplayProof
\\[10px]
  \qquad\qquad
  \AxiomC{$X,Y\turn Z$}
  \LeftLabel{(Com $\turn$)\,}
  \UnaryInfC{$Y,X\turn Z$}
  \DisplayProof
  &
  \AxiomC{$X\turn Y,Z$}
  \LeftLabel{($\turn$ Com)\,}
  \UnaryInfC{$X\turn Z,Y$}
  \DisplayProof
\\[10px]
  \qquad\qquad
  \AxiomC{$W,(X<Y)\turn Z$}
    \LeftLabel{(Grnb $\turn$)}
  \UnaryInfC{$(W,X)<Y\turn Z$}
  \DisplayProof
  &
  \AxiomC{$W\vdash (X> Y), Z$}
    \LeftLabel{($\turn$ Grnb)}
  \UnaryInfC{$W\vdash X>(Y,Z)$}
\DisplayProof
\\[5px]
\multicolumn{2}{l}{Further logical rules for $\biilldc$:}
\\[5px]
  \qquad\qquad
  \AxiomC{$A<B\turn X$}
    \LeftLabel{($\excl\turn$)\,}
  \UnaryInfC{$A\excl B\turn X$}
  \DisplayProof
  &
  \AxiomC{$X\turn A$}
  \AxiomC{$B\turn Y$}
  \LeftLabel{($\turn\excl$)\,}
  \BinaryInfC{$X<Y\turn A\excl B$}
  \DisplayProof
  \end{tabular}
\caption{$\filldc$ and $\biilldc$: display calculi for FILL and BiILL}
\label{Fig:dc}
\end{figure}

\begin{remark}
  For conciseness, we treat comma-separated structures as multisets
  and usually omit explicit use of (Ass $\turn$), ($\turn$ Ass), (Com
  $\turn$) and ($\turn$ Com). The \emph{residuated pair} and
  \emph{dual residuated pair} rules (rp) and (drp) are the
  \emph{display postulates} which give Thm.~\ref{thm:display} below.
  Our display postulates build in commutativity of comma, so 
  the two (Com) rules are derivable. If we wanted to drop
  commutativity~\cite{Cockett:Proof}, we would have to use the more
  general display postulates from~\cite{Gore:SLD}. Note that (drp) may
  create the structure $<$ which has no meaning in FILL, so we will
  return to this issue. For now, observe that proofs of even
  apparently trivial FILL-sequents such as $(p\parr q)\parr r\turn
  p,(q\parr r)$ require (drp) to `move $p$ out the way' so
  ($\turn\parr$) can be applied. Another (drp) then eliminates the $<$
  to restore $p$ to the right. The rule ($\turn$ Grnb) is the
  structural version of Grishin (b), the right hand formula of
  \eqref{eq:wdis}; the rule (Grnb $\turn$) is equivalent.
%
%
Fig.~\ref{Fig:Bierman} gives a cut-free proof of the example from Bierman~\eqref{eq:Bierman}.
\end{remark}

\begin{figure}
\footnotesize{\begin{center}$
\AxiomC{$a\turn a$}
\AxiomC{$b\turn b$}\LeftLabel{\scriptsize{($\parr\turn$)}}
\BinaryInfC{$a\parr b\turn a,b$}
\AxiomC{$c\turn c$}\LeftLabel{\scriptsize{($\parr\turn$)}}
\BinaryInfC{$(a\parr b)\parr c \turn a,b,c$} \LeftLabel{\scriptsize{\textbf{(drp)}}}
\UnaryInfC{$(a\parr b)\parr c<a \turn b,c$} \LeftLabel{\scriptsize{($\turn\parr$)}}
\UnaryInfC{$(a\parr b)\parr c<a\turn b\parr c$}
\AxiomC{$d\turn d$}\LeftLabel{\scriptsize{($\limp\turn$)}}
\BinaryInfC{$b\parr c\limp d \turn ((a\parr b)\parr c<a)>d$}
\AxiomC{$e\turn e$}\LeftLabel{\scriptsize{($\parr\turn$)}}
\BinaryInfC{$(b\parr c\limp d)\parr e \turn (((a\parr b)\parr c<a)>d),e$}\LeftLabel{\scriptsize{($\turn$ Grnb)}}
\UnaryInfC{$(b\parr c\limp d)\parr e \turn ((a\parr b)\parr c<a)>d,e$}\LeftLabel{\scriptsize{(rp)}}
\UnaryInfC{$(b\parr c\limp d)\parr e,((a\parr b)\parr c<a)\turn d,e$}\LeftLabel{\scriptsize{($\turn\parr$)}}
\UnaryInfC{$(b\parr c\limp d)\parr e,((a\parr b)\parr c<a)\turn d\parr e$}\LeftLabel{\scriptsize{(rp)}}
\UnaryInfC{$(a\parr b)\parr c<a \turn (b\parr c\limp d)\parr e>d\parr e$}\LeftLabel{\scriptsize{($\turn\limp$)}}
\UnaryInfC{$(a\parr b)\parr c<a \turn (b\parr c\limp d)\parr e\limp d\parr e$}\LeftLabel{\scriptsize{\textbf{(drp)}}}
\UnaryInfC{$(a\parr b)\parr c \turn a,(b\parr c\limp d)\parr e\limp d\parr e$}
\DisplayProof
$\end{center}}
\caption{The cut-free $\filldc$-derivation of the example from Bierman.}
\label{Fig:Bierman}
\end{figure}

\begin{theorem}[Display Property]\label{thm:display}
 For every structure $Z$ which is an antecedent (resp. succedent)
 part of the sequent $X \turn Y$, there is a sequent $Z \turn Y'$
 (resp. $X' \turn Z$) obtainable from $X \turn Y$ using only (rp) and
 (drp), thereby displaying the $Z$ as the whole of one side.
\end{theorem}

\begin{theorem}[Cut-Admissibility]\label{thm:cutad}
  From cut-free $\biilldc$-derivations of $X \turn A$ and 
  $A \turn Y$ there is an effective procedure to obtain
  a cut-free $\biilldc$-derivation of $X
  \turn Y$.
\end{theorem}

\begin{proof}
  $\biilldc$ obeys Belnap's conditions for
  cut-admissibility~\cite{Belnap82JPL}: see
  App.~\ref{app:display}.
\end{proof}

\subsection{Semantics}

\begin{definition}\label{Def:FILLcat}
A \emph{FILL-category} is a category equipped with
\begin{itemize}
\item a symmetric monoidal closed structure $(\otimes,I,\limp)$
\item a symmetric monoidal structure $(\parr,\bot)$
\item a natural family of \emph{weak distributivity} arrows $A\otimes(B\parr C)\to(A\otimes B)\parr C$.
\end{itemize}
A \emph{BiILL-category} is a FILL-category where the $\parr$ bifunctor has a
\emph{co-closure} $\excl$, so there is a natural isomorphism between arrows $A
\to B\parr C$ and $A\excl B\to C$.

\end{definition}

\begin{definition}\label{def:free}
  The \emph{free} FILL- (resp. BiILL-) category has FILL- (resp.
  BiILL-) formulae as objects and the following arrows (quotiented by
  certain equations) where we are given objects $A, A',A'',B,B'$ and
  arrows $f:A\to A',f':A'\to A'',g:B\to B'$,
  $(\binop,K)\in\{(\otimes,I),(\parr,\bot)\}$, and where the
  co-closure arrows exist in the free BiILL-category only:
\begin{description}
\item[Category:] 
  $\xymatrix{A \ar[r]^{id} & A}
  \qquad
  \xymatrix{A \ar[r]^-{f'\circ f} & A''}$
\item[Symmetric Monoidal:]
  $\xymatrix{A\binop B \ar[r]^-{f\binop g} & A'\binop B'}
   \qquad
  \xymatrix{(A\binop B)\binop C \ar@<0.5ex>[r]^{\alpha} & A\binop(B\binop C)
  \ar@<0.5ex>[l]^{\alpha^{-1}}}
  $   
\item[\hphantom{Symmetric Monoidal:}]
  $
  \xymatrix{K\binop A \ar@<0.5ex>[r]^-{\lambda} & A
  \ar@<0.5ex>[l]^-{\lambda^{-1}}}
  \qquad
  \xymatrix{A\binop K \ar@<0.5ex>[r]^-{\rho} & A \ar@<0.5ex>[l]^-{\rho^{-1}}}
  \qquad
  \xymatrix{A\binop B \ar[r]^{\gamma} & B\binop A}
  $
\item[Closed:] 
  $\xymatrix{A\limp B \ar[r]^-{A\limp g} & A\limp B'}
   \qquad
   \xymatrix{(A\limp B)\otimes A \ar[r]^-{\varepsilon} & B}
   \qquad
  \xymatrix{A \ar[r]^-{\eta} & B\limp A\otimes B}$
\item[Weak Distributivity:]
  $\xymatrix{A\otimes(A'\parr A'') \ar[r]^{\omega} & (A\otimes A')\parr A''}$
\item[Co-Closed:] 
  $\xymatrix{A\excl B \ar[r]^-{f\excl B} & A'\excl B}
   \qquad
   \xymatrix{A\parr B\excl A \ar[r]^-{\varepsilon} & B}
   \qquad
   \xymatrix{A \ar[r]^-{\eta} & B\parr(A\excl B)}$
\end{description}

We will suppress explicit reference to
the associativity and symmetry arrows.
\end{definition}

\begin{definition}\label{Def:valid}
  A FILL- (resp. BiILL-) sequent $X\turn Y$ is \emph{satisfied} by a
  FILL- (resp.\ BiILL-) category if, given any valuation of its
  propositional variables as objects, there exists an arrow
  $I\to\tau^a(X)\limp\tau^s(Y)$.
  It is FILL- (resp. BiILL-) \emph{valid} if
  it is satisfied by all such categories. In fact, we only need to
  check the free categories under their generic valuations.
\end{definition}

\begin{remark}
Those familiar with categorical logic will note that our use of category theory
here is rather shallow, looking only at whether hom-sets are populated,
and not at the rich structure of equivalences between proofs that
categorical logic supports. This is
an adequate basis for this work because the question of FILL-validity
alone has proved so vexed.
\end{remark}

\begin{theorem}\label{thm:dc=biill}
$\biilldc$ (Fig.~\ref{Fig:dc}) is \emph{sound} and \emph{cut-free complete} for
BiILL-validity.
\end{theorem}
\begin{proof}
$\biilldc$-proof rules and the arrows of the free BiILL-category are
interdefinable.
\end{proof}

\begin{corollary}\label{cor:filldc_complete}
The display calculus $\filldc$ is cut-free complete for FILL-validity.
\end{corollary}
\begin{proof}
Because BiILL-categories are FILL-categories, and $\biilldc$ proofs of
FILL-sequents are $\filldc$ proofs.
\end{proof}

We will return to the question of \emph{soundness} for $\filldc$ in
Sec.~\ref{sec:conserv}.




\section{Deep Inference and Proof Search}
\label{sec:nested-sequent}

We now present a refinement of the display calculus $\biilldc$,
in the form of a nested sequent calculus, that is more suitable
for proof search. A nested sequent
is essentially just a structure in display calculus, but presented in a more sequent-like notation. 
This change of notation allows us to present the proof systems much more concisely. 
The proof system we are interested in is the deep inference
system in Sec.~\ref{sec:deep}, but we shall first present an intermediate system, 
$\biillsn$, which is closer to display calculus, and which eases the proof of
correspondence between the deep inference calculus and the display calculus for $\biill.$


\subsection{The Shallow Inference Calculus}
\label{sec:shallow}

The syntax of nested sequents is given by the grammar below
where $A_i$ and $B_j$ are formulae. 
$$
S~~T ::= S_1,\dots,S_k,A_1,\dots,A_m \seq B_1,\dots, B_n, T_1, \dots,T_l
$$
We use $\Gamma$ and $\Delta$ for multisets of formulae and use
$P$, $Q$, $S$, $T$, $X$, $Y$, etc., for sequents, and $\Sc$, $\Xc$, etc., for
multisets of sequents and formulae. The empty multiset is 
$\sunit$ (`dot'). 


A nested sequent can naturally be represented as a tree structure as follows.
The nodes of the tree are traditional two-sided sequents (i.e., pairs of multisets). The edges
between nodes are labelled with either a $-$, denoting nesting to the left of the sequent
arrow, or a $+$, denoting nesting to the right of the sequent arrow. For example,
the nested sequent below
can be visualised as the tree in Fig.~\ref{fig:partition} (i):
\begin{equation}
  \label{eq:1}
(e,f \seq g), (p, (u,v \seq x,y) \seq q,r), a, b \seq c, d, (\cdot \seq s)
\end{equation}
A display sequent can be seen as a nested sequent, where $\turn$,
$>$ and $<$ are all replaced by $\seq$ and the unit $\Phi$ is represented by the empty multiset.
The definition of a nested sequent incorporates implicitly the associativity and
commutativity of comma, and the effects of its unit, via the multiset structure.

\begin{definition}\label{def:tau_nested}
Following Def.~\ref{def:tau}, we can translate nested sequents into equivalence
classes of BiILL-formulae (modulo associativity, commutativity, and unit laws)
via
\emph{$\tau$-translations}:
$$
\begin{array}{l}
\tau^a(S_1,\dots,S_k, A_1,\dots,A_m \seq B_1,\dots,B_n, T_1,\dots,T_l) \\
= 
(\tau^a(S_1) \ltens \cdots \ltens \tau^a(S_k) \ltens A_1 \ltens \cdots \ltens A_m) 
\excl
(B_1 \parr \cdots \parr B_n \parr \tau^s(T_1) \parr \cdots \parr \tau^s(T_l))
\\[1em]
\tau^s(S_1,\dots,S_k, A_1,\dots,A_m \seq B_1,\dots,B_n, T_1,\dots,T_l) \\
= 
(\tau^a(S_1) \ltens \cdots \ltens \tau^a(S_k) \ltens A_1 \ltens \cdots \ltens A_m) 
\limp
(B_1 \parr \cdots \parr B_n \parr \tau^s(T_1) \parr \cdots \parr \tau^s(T_l)).
\end{array}
$$
The translations
$\tau^a$ and $\tau^s$ differ only in their translation of the sequent symbol
$\seq$ to $\limp$ and $\excl$ respectively. Where $m = 0$, $A_1\ltens\cdots
\ltens A_m$ translates to $I$, and similarly $B_1\parr\cdots\parr B_n$
translates to
$\bot$ when $n =0$. These translations each extend to a map from multisets of nested
sequents and formulae to formulae: $\tau^a$ (resp. $\tau^s$) acts on each
sequent as above, leaves formulae unchanged, and connects the resulting formulae
with $\otimes$ (resp. $\parr$). Empty multisets are mapped to $I$ (resp.
$\bot$).
\end{definition}



A {\em context} is either a `hole' $[~]$, called the {\em empty context}, 
or a sequent where exactly one node has been replaced by a hole $[~]$. 
Contexts are denoted by $X[~]$.
We write $X[S]$ to denote a sequent
resulting from replacing the hole $[~]$ in $X[~]$ with
the sequent $S$.
A non-empty context $X[~]$ is {\em positive} if the hole $[~]$ occurs immediately
to the right of a sequent arrow $\seq$, and {\em negative} otherwise.
This simple definition of polarities of a context is made possible by the use of the same symbol $\seq$
to denote the structural counterparts of $\limp$ and $\excl$. 
As we shall see in Sec.~\ref{sec:deep}, this overloading of $\seq$ allows a
presentation of deep inference rules that ignores context polarity. 



The shallow inference system $\biillsn$ for $\biill$ is given
in Fig.~\ref{fig:fbills}.
The main difference from $\biilldc$ is that
we allow multiple-conclusion logical rules. This implicitly builds the Grishin (b) rules into 
the logical rules (see App. D).

\begin{theorem}
\label{thm:biill-sn-eq-dc}
A formula is cut-free $\biillsn$-provable iff it is cut-free $\biilldc$-provable. 
\end{theorem}

\begin{corollary}
\label{cor:cut-elim-biillsn}
The cut rule is admissible in $\biillsn.$
\end{corollary}

Just as in display calculus (Thm. \ref{thm:display}), the display property holds
for $\biillsn$. 

\begin{proposition}[Display property]
\label{prop:display}
Let $X[~]$ be a positive (negative) context. For every $\Sc$,
there exists $\Tc$ such that $\Tc \seq \Sc$ (respectively $\Sc \seq \Tc$) 
is derivable from $X[\Sc]$ 
using only the structural rules from $\{ \drp_1, drp_2, \rp_1,
rp_2\}.$
Thus $\Sc$ is ``displayed'' in $\Tc\seq\Sc$ ($\Sc\seq\Tc$).
\end{proposition}

\begin{figure}[t]
{\footnotesize
$$
\begin{array}{ccc}
\mbox{$
\xymatrix@C=0pt@R=12pt{
 & {a, b \seq c,d} \ar[dl]_{-} \ar[d]^{-} \ar[dr]^{+}\\
{e,f \seq g} & {p \seq q,r} \ar[d]^{-} & {\cdot \seq s} \\
 & {u, v \seq x,y }
}
$}
&
\mbox{
$
\xymatrix@C=0pt@R=12pt{
 & {a \seq c} \ar[dl]_{-} \ar[d]^{-} \ar[dr]^{+}\\
{e \seq g} & {p \seq } \ar[d]^{-} & {\cdot \seq \cdot} \\
 & { u \seq x }
}
$}
&
\mbox{
$
\xymatrix@C=0pt@R=12pt{
 & {b \seq d} \ar[dl]_{-} \ar[d]^{-} \ar[dr]^{+}\\
{f \seq \cdot} & {\cdot \seq q,r} \ar[d]^{-} & {\cdot \seq s} \\
 & {v \seq y }
}
$} \\
(i) & (ii) & (iii)
\end{array}
$$
}
\caption{A tree representation of a nested sequent (i), and its partitions (ii and iii). }
\label{fig:partition}
\end{figure}
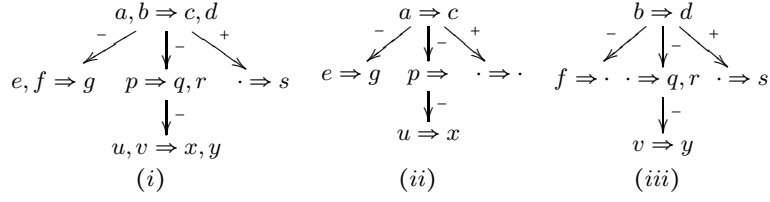

\begin{figure}[t]
{\small
Cut and identity: 
$
\qquad
\vcenter{
\infer[id]
{p \seq p}{}
}
\qquad
\vcenter{
\infer[cut]
{\Sc,\Tc \seq \Sc', \Tc'}
{\Sc \seq \Sc', A & A, \Tc \seq \Tc'}
}
$

Structural rules:
$$
\begin{array}{ccc}
\mbox{$
\infer[\drp_1]
{(\Sc \seq \Tc) \seq \Tc'}
{\Sc \seq \Tc, \Tc'}
$}
&
\mbox{$
\infer[\rp_1]
{\Sc \seq (\Tc \seq \Tc')}
{
\Sc, \Tc \seq \Tc'
}
$}
& 
\mbox{$\infer[gl]
{(\Sc, \Tc \seq \Sc') \seq \Tc'}
{(\Sc \seq \Sc'), \Tc \seq \Tc'}
$} \\ \\
\mbox{$
\infer[drp_2]
{\Sc \seq \Tc, \Tc'}
{(\Sc \seq \Tc) \seq \Tc'}
$}
& 
\mbox{$
\infer[rp_2]
{\Sc, \Tc \seq \Tc'}
{\Sc \seq (\Tc \seq \Tc')}
$}
& 
\mbox{$
\infer[gr]
{\Sc \seq (\Sc' \seq \Tc', \Tc)}
{
 \Sc \seq (\Sc' \seq \Tc'), \Tc
}
$}
\end{array}
$$


Logical rules:
$$
\infer[\punit_l]
{\punit \seq \sunit}
{}
\qquad
\infer[\punit_r]
{\Sc \seq \Tc, \punit}
{\Sc \seq \Tc}
\qquad
\infer[\tunit_l]
{\Sc, \tunit \seq \Tc}
{\Sc \seq \Tc}
\qquad
\infer[\tunit_r]
{\sunit \seq \tunit}{}
$$
$$
\infer[\ltens_l]
{\Sc, A \ltens B \seq \Tc}
{\Sc, A, B \seq \Tc}
\qquad 
\infer[\ltens_r]
{\Sc, \Sc' \seq A \ltens B, \Tc, \Tc'}
{\Sc \seq A, \Tc & \Sc' \seq B, \Tc'}
$$
$$
\infer[\parr_l]
{\Sc, \Sc', A \parr B \seq \Tc, \Tc'}
{\Sc, A \seq \Tc & \Sc', B \seq \Tc'}
\qquad
\infer[\parr_r]
{\Sc \seq A \parr B, \Tc}
{\Sc \seq A, B, \Tc}
$$
$$
\infer[\limp_l]
{\Sc, \Sc', A \limp B \seq \Tc, \Tc'}
{
\Sc \seq A, \Tc
&
\Sc', B \seq \Tc'
}
\qquad
\infer[\limp_r]
{\Sc \seq \Tc, A \limp B}
{\Sc \seq \Tc, (A \seq B)}
$$
$$
\infer[\excl_l]
{\Sc, A \excl B \seq \Tc}
{
 \Sc, (A \seq B) \seq \Tc
}
\qquad
\infer[\excl_r]
{\Sc, \Sc' \seq A \excl B, \Tc, \Tc'}
{
 \Sc \seq A, \Tc & 
 \Sc', B \seq \Tc'
}
$$
}
\caption{The shallow inference system $\biillsn$, where $gl$ and $gr$
  capture Grishin (b).}
\label{fig:fbills}
\end{figure}



\subsection{The Deep Inference Calculus}
\label{sec:deep}

\begin{figure}[t]
{\small


Propagation rules:

$$
\infer[pl_1]
{X[\Sc, A \seq (\Sc' \seq \Tc'), \Tc]}
{
 X[\Sc \seq (A, \Sc' \seq \Tc'), \Tc]
}
\qquad
\infer[pr_1]
{X[(\Sc \seq \Tc), \Sc' \seq A, \Tc']}
{X[(\Sc \seq \Tc, A), \Sc' \seq \Tc']}
$$
$$
\infer[pl_2]
{X[\Sc, (\Sc', A \seq \Tc') \seq \Tc]}
{X[\Sc, A, (\Sc' \seq \Tc') \seq \Tc]}
\qquad
\infer[pr_2]
{X[\Sc \seq \Tc, (\Sc' \seq \Tc', A)]}
{X[\Sc \seq \Tc, A, (\Sc' \seq \Tc')]}
$$

Identity and logical rules: In branching rules, $X[~] \in X_1[~] \merge X_2[~]$,
$\Sc \in \Sc_1 \merge \Sc_2$ and $\Tc \in \Tc_1 \merge \Tc_2$.
$$
\infer[id^d]
{X[\Uc, p \seq p, \Vc]}
{\mbox{$X[~]$, $\Uc$ and $\Vc$ are hollow.}}
\qquad
\infer[\punit_l^d]
{X[\punit, \Uc \seq \Vc]}
{\mbox{$X[~]$, $\Uc$ and $\Vc$ are hollow. }}
\qquad
\infer[\punit_r^d]
{X[\Sc \seq \Tc, \punit]}
{X[\Sc \seq \Tc]}
$$
$$
\infer[\tunit_l^d]
{X[\Sc, \tunit \seq \Tc]}
{X[\Sc \seq \Tc]}
\qquad
\infer[\tunit_r^d]
{X[\Uc \seq \tunit, \Vc]}
{\mbox{$X[~]$, $\Uc$ and $\Vc$ are hollow. }}
$$
$$
\infer[\ltens_l^d]
{X[\Sc, A \ltens B \seq \Tc]}
{X[\Sc, A, B \seq \Tc]}
\qquad
\infer[\ltens_r^d]
{X[\Sc \seq A \ltens B, \Tc]}
{X_1[\Sc_1 \seq A, \Tc_1] & X_2[\Sc_2 \seq B, \Tc_2]}
$$
$$
\infer[\limp_l^d]
{X[\Sc, A \limp B \seq \Tc]}
{
X_1[\Sc_1 \seq A, \Tc_1] & X_2[\Sc_2, B \seq \Tc_2]
}
\qquad
\infer[\limp_r^d]
{X[\Sc \seq \Tc, A \limp B]}
{X[\Sc \seq \Tc, (A \seq B)]}
$$
$$
\infer[\parr_l^d]
{X[\Sc, A \parr B \seq \Tc]}
{
 X_1[\Sc_1, A \seq \Tc_1]
 &
 X_2[\Sc_2, B \seq \Tc_2]
}
\qquad
\infer[\parr_r^d]
{X[\Sc \seq A \parr B, \Tc]}
{X[\Sc \seq A, B, \Tc]}
$$
$$
\infer[\excl_l^d]
{X[\Sc, A \excl B \seq \Tc]}
{X[\Sc, (A \seq B) \seq \Tc]}
\qquad
\infer[\excl_r^d]
{X[\Sc \seq A \excl B, \Tc]}
{X_1[\Sc_1 \seq A, \Tc_1] 
&
 X_2[\Sc_2, B \seq \Tc_2]
}
$$
}
\caption{The deep inference system $\biilldn$. }
\label{fig:fbilld}
\end{figure}

A deep inference rule can be applied to any sequent within a nested sequent. This poses a problem
in formalising context splitting rules, e.g., $\ltens$ on the right. To be sound, we need
to consider a context splitting that splits an entire tree of sequents, as formalised next. 



Given two sequents $X_1$ and $X_2$, their {\em merge set} $X_1
\merge X_2$ is defined inductively as:
\begin{center}
\begin{tabular}[c]{l@{\extracolsep{2cm}}cl}
 \multicolumn{3}{l}{
  $X_1 \merge X_2 = \{~ 
   (\Gamma_1, \Gamma_2, Y_1, \dots, Y_m \seq \Delta_1, \Delta_2, Z_1,
    \dots, Z_n) ~\mid$
 }
\\
  & & $X_1 = (\Gamma_1, P_1, \dots, P_m \seq \Delta_1, Q_1,\dots,Q_n)$
  \mbox{ and }
\\
  & & 
  $X_2 = (\Gamma_2, S_1, \dots, S_m \seq \Delta_2, T_1, \dots, T_n)$
 \mbox{ and }
\\
 & &
  $Y_i \in P_i \merge S_i$ for $1 \leq i \leq m$ 
  and $Z_j \in Q_j \merge T_j$ for $1 \leq j \leq n$
  $~\}$
\end{tabular}
\end{center}


Note that the merge set of two sequents may not always be defined
since mergeable sequents need to have the same structure.
Note also that, because there can be more than one way to enumerate
elements of a multiset in the left/right hand side of a sequent, 
the result of the merging of two nested sequents is a set, rather
than a single nested sequent. 
When
$X \in X_1 \merge X_2$, we say that $X_1$ and $X_2$ are a {\em
  partition} of $X.$ Fig.~\ref{fig:partition} (ii) and (iii) show a partitioning
of the nested sequent~\eqref{eq:1}
in the tree representation.
Note that the partitions (ii) and (iii) must have the same tree
structure as the original sequent (i).


Given two contexts
$X_1[~]$ and $X_2[~]$ 
their merge set $X_1[~] \merge X_2[~]$
is defined as follows:

  \begin{tabular}[c]{lcl}
  \multicolumn{3}{l}{  
     If $X_1[~] = [~]$ and $X_2[~] = [~]$ 
     then $X_1[~] \merge X_2[~] = \{ [~] \}$
   }
\\
    If $X_1[~]$ & $=$ &
          $(\Gamma_1, Y_1[~], P_1,\dots,P_m \seq \Delta_1, Q_1,\dots,Q_n)$
    and
\\
    \hphantom{If} $X_2[~]$ & $=$ &
          $(\Gamma_2, Y_2[~], S_1,\dots,S_m \seq \Delta_2, T_1,\dots,T_n)$
   then
\\
   $X_1[~] \merge X_2[~]$ & $=$ &
      $\{~ (\Gamma_1, \Gamma_2, Y[~], U_1,\dots,U_m \seq \Delta_1,\Delta_2,
          V_1,\dots,V_n)~~ \mid$
\\
   & &
        $
        \begin{array}{l} 
         Y[~] \in Y_1[~]\merge Y_2[~] 
        \mbox{ and }
        U_i \in P_i \merge S_i 
        \mbox{ for } 
        1 \leq i \leq m
        \mbox{ and } \\
        V_j \in Q_j \merge T_j
        \mbox{ for } 
        1 \leq j \leq n 
       ~\}
        \end{array}$
\\
    If $X_1[~]$ & $=$ & $(\Gamma_1, P_1,\dots,P_m \seq \Delta_1, Y_1[~], Q_1,\dots,Q_n)$
    and
    \\
    \hphantom{If} $X_2[~]$ & $=$ & $(\Gamma_2, S_1,\dots,S_m \seq \Delta_2, Y_2[~], T_1,\dots,T_n)$
    then
    \\
    $X_1[~] \merge X_2[~]$ & $=$ & 
    $\{~
      (\Gamma_1, \Gamma_2, U_1,\dots,U_m \seq \Delta_1,\Delta_2, Y[~],
       V_1,\dots,V_n)$ $\mid$
    \\
     & &
    $\begin{array}{l}
     Y[~] \in Y_1[~]\merge Y_2[~] \mbox{ and }
     U_i \in P_i \merge S_i \mbox{ for } 1 \leq i \leq m \mbox{ and }\\
     V_j \in Q_j \merge T_j$ for $1 \leq j \leq n ~\}
     \end{array}
    $
  \end{tabular}

If $X[~] = X_1[~] \merge X_2[~]$ we say $X_1[~]$ and $X_2[~]$
are a {\em partition} of $X[~]$.



We extend the notion of a merge set between multisets of formulae and sequents
as follows. Given 
$
\Xc = \Gamma \cup \{X_1,\dots,X_n\}
$
and
$
\Yc = \Delta \cup \{Y_1,\dots,Y_n\} 
$
their merge set contains all multisets of the form:
$
\Gamma \cup \Delta \cup \{Z_1,\dots,Z_n\}
$
where $Z_i \in X_i \merge Y_i.$

A nested sequent $X$ (resp. a context $X[~]$) 
is said to be {\em hollow} iff it contains no occurrences of formulae.
For example, 
$(\sunit \seq \sunit) \seq (\sunit \seq [~]), (\sunit \seq \sunit)$ is a hollow context.

The deep inference system for $\biill$, called $\biilldn$, 
is given in Fig.~\ref{fig:fbilld}. 
Fig.~\ref{fig:bierman-deep} shows a cut-free derivation of Bierman's example in $\biilldn.$

\begin{figure}
{\footnotesize
$$
\infer[\limp_r^d]
{(a \parr b) \parr c \seq a, \mathbf{(b \parr c \limp d) \parr e \limp d \parr e} }
{
 \infer[\parr_r^d]
 { (a \parr b) \parr c \seq a, ((b \parr c \limp d) \parr e \seq \mathbf{d \parr e})}
 {
  \infer[\parr_l^d]
  { (a \parr b) \parr c \seq a, (\mathbf{(b \parr c \limp d) \parr e} \seq d, e)}
  {
   \infer[\limp_l^d]
   {(a \parr b) \parr c \seq a, (\mathbf{b \parr c \limp d} \seq d)} 
   {
    \infer[\parr_r^d]
    {(a \parr b) \parr c \seq a, (\sunit \seq \mathbf{b \parr c})} 
    {
     \infer[\parr_l^d]
     {\mathbf{(a \parr b) \parr c} \seq a, (\sunit \seq b, c)} 
     {
      \infer[\parr_l^d]
      {\mathbf{a \parr b} \seq a, (\sunit \seq b)}
      {
        \infer[id^d]
        {a \seq a, (\sunit \seq \sunit)}
        { }
        &
        \infer[pl_1]
        {b \seq (\sunit \seq b)}
        {
           \infer[id^d]
           {\sunit \seq (b \seq b)}{}
        }
      }
      &
      \infer[pl_1]
      {c \seq (\sunit \seq c)}
      {
       \infer[id^d]
       {\sunit \seq (c \seq c)}
       {}
      }
     }
    }
    &
    \!\!
    \infer[id^d]
    {\sunit \seq (d \seq d)}
    {}
   }
   &
   \infer[id^d]
   { \sunit \seq (e \seq e)}
   {}
  }
 } 
}
$$
}
\caption{A cut-free derivation of Bierman's example in $\biilldn$.}
\label{fig:bierman-deep}
\end{figure}



\subsection{The Equivalence of the Deep and Shallow Nested Sequent Calculi}
\label{sec:equiv}



From $\biilldn$ to $\biillsn$, it is enough to show that
every deep inference rule is {\em cut-free derivable} in $\biillsn$.
For the identity and the constant rules, this follows from the fact that hollow structures
can be weakened away, as they add nothing to provability (see App. E).
For the other logical rules, a key idea to their soundness 
is that the context splitting operation is derivable in $\biillsn.$ 
This is a consequence of the following lemma (see App. E.1).

\begin{lemma}
\label{lm:dist}
The following rules are derivable in $\biillsn$ without cut:
$$
\infer[dist_l]
{(\Xc_1, \Xc_2  \seq \Yc_1,\Yc_2), \Uc \seq \Vc}
{
(\Xc_1 \seq \Yc_1), (\Xc_2 \seq \Yc_2), \Uc \seq \Vc
}
\qquad
\infer[dist_r]
{\Uc \seq \Vc, (\Xc_1, \Xc_2  \seq \Yc_1,\Yc_2)}
{
 \Uc \seq \Vc, (\Xc_1 \seq \Yc_1), (\Xc_2 \seq \Yc_2)
}
$$
\end{lemma}
Intuitively, these rules embody
the weak distributivity formalised by the Grishin (b) rule. 

\begin{lemma}
\label{lm:merge}
If $\Xc \in \Xc_1 \merge \Xc_2$ then the rules below are cut-free derivable in $\biillsn$:
$$
\infer[m_l]
{\Xc, \Uc \seq \Vc}
{
\Xc_1, \Xc_2, \Uc \seq \Vc
}
\qquad
\infer[m_r]
{\Uc \seq \Vc, \Xc}
{
\Uc \seq \Vc, \Xc_1, \Xc_2
}
$$
\end{lemma}
\begin{proof}
This follows straightforwardly from Lem.~\ref{lm:dist}.
\end{proof}

\begin{lemma}
\label{lm:context-merge}
Suppose $X[~] \in X_1[~] \merge X_2[~]$ and suppose there exists $Y[~]$
such that for any $\Uc$ 
and any $\rho \in \{drp_1, drp_2, rp_1, rp_2\}$,
the figure below left is a valid inference
rule in $\biillsn$:
\[
\qquad
\qquad\qquad
\infer[\rho]
{X[\Uc]}
{Y[\Uc]}
\qquad\qquad
\infer[\rho]
{X_1[\Uc]}
{Y_1[\Uc]}
\qquad
\qquad
\infer[\rho]
{X_2[\Uc]}
{Y_2[\Uc]}
\]
Then there exists $Y_1[~]$ and $Y_2[~]$ such that
$Y[~] \in Y_1[~] \merge Y_2[~]$ and the second and the third figures above 
are also valid instances of $\rho$ in $\biillsn$.
\end{lemma}
\begin{proof}
This follows from the fact that $X[~]$, $X_1[~]$ and $X_2[~]$ have exactly the same
nested structure, so whatever display rule applies to one also applies to the others. 
\end{proof}

\begin{theorem}
\label{thm:soundness}
If a sequent $X$ is provable in $\biilldn$ then it is cut-free provable in $\biillsn.$
\end{theorem}
\begin{proof}
We show that every rule of $\biilldn$ is cut-free derivable in $\biillsn.$
We show here a derivation of the rule $\limp^d_l$; the rest can be proved similarly.
So suppose the conclusion of the rule is $X[\Sc, A \limp B \seq \Tc]$, and the
premises are $X_1[\Sc_1 \seq A, \Tc_1]$ and $X_2[\Sc_2, B \seq \Tc_2]$, where
$X[~] \in X_1[~]\merge X_2[~]$, $\Sc \in \Sc_1 \merge \Sc_2$ and
$\Tc \in \Tc_1 \merge \Tc_2.$
There are two cases to consider, depending on whether $X[~]$ is positive or negative.
We show here the former case, as the latter case is similar. 
Prop.~\ref{prop:display} entails that $X[\Sc, A \limp B \seq \Tc]$ is display equivalent to
$\Uc \seq (\Sc, A \limp B \seq \Tc)$ for some $\Uc$. By Lem.~\ref{lm:context-merge}, we have
$\Uc_1$ and $\Uc_2$ such that
$\Uc \in \Uc_1 \merge \Uc_2$, and $(\Uc_1 \seq \Vc)$ and $(\Uc_2 \seq \Vc)$ 
are display equivalent to, respectively, $X_1[\Vc]$ and $X_2[\Vc]$, for any $\Vc.$
The derivation of $\limp^d_l$ in $\biillsn$ is thus constructed as follows:
$$
\infer[\mbox{Prop. \ref{prop:display}}]
{X[\Sc, A \limp B \seq \Tc]}
{
 \infer[rp_1]
 {\Uc \seq (\Sc, A \limp B \seq \Tc)}
 {
  \infer[ml;ml;mr]
  {\Uc, \Sc, A \limp B \seq \Tc}
  {
   \infer[\limp_l]
   {\Uc_1, \Uc_2, \Sc_1, \Sc_2, A \limp B \seq \Tc_1, \Tc_2}
   {
    \infer[rp_2]
    {\Uc_1, \Sc_1 \seq A, \Tc_1}
    {
     \infer[\mbox{Lem. \ref{lm:context-merge}}]
     {\Uc_1 \seq (\Sc_1 \seq A, \Tc_1)}
     {X_1[\Sc_1 \seq A, \Tc_1]}
    }
    & 
    \infer[rp_2]
    {\Uc_2, \Sc_2, B \seq \Tc_2}
    {
     \infer[\mbox{Lem. \ref{lm:context-merge}}]
     {\Uc_2 \seq (\Sc_2, B \seq \Tc_2)}
     {
      X_2[\Sc_2, B \seq \Tc_2]
     }
    }
   }
  }
 }
}
$$
\end{proof}



The other direction of the equivalence is proved by a permutation argument: we first add the structural rules to 
$\biilldn$, then we show that these structural rules permute up over all (non-constant) 
logical rules of $\biilldn.$ Then when the structural rules appear just below the $id^d$ or
the constant rules, they become redundant. There are quite a number of cases to consider, but
they are not difficult once one observes the following property of $\biilldn$: 
in every rule, every context in the premise(s) has the same tree structure 
as the context in the conclusion of the rule. This observation takes care of permuting
up structural rules that affect only the context. The non-trivial cases are those where
the application of the structural rules changes the sequent where the logical rule is applied. 
We illustrate a case in the following lemma. The detailed proof 
can be found in App. E.2.


\begin{lemma}
\label{lm:permute}
The rules $drp_1$, $rp_1$, $drp_2$, $rp_2$, $gl$, and $gr$ permute up over all logical rules of $\biilldn.$
\end{lemma}
\begin{proof}
({\em Outline}) 
We illustrate here 
a non-trivial interaction between a structural rule
and $\limp_l$, where the conclusion sequent of $\limp_l$ is changed by 
that structural rule. The other non-trivial cases follow the same pattern, i.e., propagation
rules are used to move the principal formula to the required structural context. 

{\footnotesize
$$
\infer[rp_1]
{\Sc, C \limp B \seq (\Tc \seq \Uc)}
{
\infer[\limp_l]
{\Sc, C \limp B, \Tc \seq \Uc}
{
 {\Sc_1, \Tc_1 \seq C, \Uc_1}
&
 {\Sc_2, \Tc_2 , B \seq  \Uc_2}
}
}
\quad
\leadsto
\quad
\infer[pl_1]
{\Sc, C \limp B \seq (\Tc \seq \Uc)}
{
 \infer[\limp_l]
 {\Sc \seq (C \limp B, \Tc \seq \Uc)}
 {
  \infer[rp_1] 
  {\Sc_1 \seq (\Tc_1 \seq C, \Uc_1)}
  {\Sc_1, \Tc_1 \seq C, \Uc_1}
  &
  \infer[rp_1]
  {\Sc_2 \seq (\Tc_2, B \seq \Uc_2)}
  {\Sc_2, \Tc_2, B \seq \Uc_2}
 }
}
$$
}
\end{proof}

\begin{theorem}
If a sequent $X$ is cut-free $\biillsn$-derivable then it is also $\biilldn$-derivable.
\end{theorem}

\begin{corollary}\label{cor:dc=dn}
A formula is cut-free $\biilldc$-derivable iff it is $\biilldn$-derivable.
\end{corollary}



\section{Separation, Conservativity, and Decidability}
\label{sec:conserv}



In this section we return our attention to the relationship between our calculi
and the categorical semantics (Defs.~\ref{Def:FILLcat} and~\ref{def:free}).
Def.~\ref{def:tau_nested} gave a translation of nested sequents to formulae;
we can hence define validity for nested sequents.

\begin{definition}
A nested sequent $S$ is BiILL-\emph{valid} if there is an arrow $I\to\tau^s(S)$
in the free BiILL-category.

A nested sequent is a (nested) \emph{FILL-sequent} if it has no nesting of
sequents on the left of $\seq$, and no occurrences of $\excl$ at all. The formula
translation of Def.~\ref{def:tau_nested} hence maps FILL-sequents to
FILL-formulae. Such a sequent $S$ is FILL-valid if there is
an arrow $I\to\tau^s(S)$ in the free FILL-category.
\end{definition}

The calculus $\biilldn$ enjoys a `separation' property
between the FILL fragment using only $\punit$, $\tunit$,
$\ltens$, $\parr$, and $\limp$ and the dual fragment
using only $\punit$, $\tunit$, $\ltens$, $\parr$, $\excl$.
Let us define $\filldn$ as the proof system obtained from
$\biilldn$ by restricting to FILL-sequents and removing the rules
$pr_1$, $pl_2$, $\excl_l^d$ and $\excl_r^d$.

\begin{theorem}[Separation]
\label{thm:separation}
Nested FILL-sequents are $\filldn$-provable iff they are
$\biilldn$-provable.
\end{theorem} 
\begin{proof}
  One direction, from $\filldn$ to $\biilldn$, is easy. The other
  holds because every sequent in a $\biilldn$ derivation
  of a FILL-sequent is also a FILL-sequent. 
\end{proof}

Thm.~\ref{thm:separation} tells us that every deep inference proof of a
FILL-sequent is entirely constructed from FILL-sequents,
each with a $\tau$-translation to FILL-formulae. This contrasts 
with display
calculus proofs, which must
introduce the FILL-untranslatable $<$ even for simple theorems.
By separation, and the equivalence of $\biilldc$ and $\biilldn$
(Cor.~\ref{cor:dc=dn}), the conservativity of BiILL over FILL reduces to
checking the soundness of each rule of $\filldn$.

\begin{lemma}\label{lem:basic_arr}
An arrow $A\otimes B\to C$ exists in the free FILL-category iff an arrow $A\to
B\limp C$ exists. Further, arrows of the following types exist for all formulae
$A,B,C$:
\begin{enumerate}
\item $A\limp B\limp C \;\to\; A\otimes B\limp C$ and
  $A\otimes B\limp C \;\to\; A\limp B\limp C$
\item $(A\limp B)\parr C \;\to\; A\limp B\parr C$.
\end{enumerate}
\end{lemma}

In the proofs below we will abuse notation by omitting explicit reference to
$\tau^a$ and $\tau^s$, writing $\Gamma_1\limp\Delta_1$ for $\tau^a(\Gamma_1)\limp
\tau^s(\Delta_1)$ for example.

\begin{lemma}\label{lem:ctxt}
Let $X[~]$ be a positive FILL-context. If there exists an arrow $f:
\tau^s(S)\to\tau^s(T)$ in the free FILL-category then there also exists an arrow
$\tau^s(X[S])\to\tau^s(X[T])$.
Hence if $X[S]$ is FILL-valid then so is $X[T]$.
\end{lemma}

\begin{lemma}\label{lem:hollow}
Given a multiset $\Vc$ of \emph{hollow} FILL-sequents, there exists an
arrow $\bot\to\tau^s(\Vc)$ in the free FILL-category.
\end{lemma}
\begin{proof}
We will prove this for a single sequent first, by induction on its size. The
base case is the sequent $\cdot\seq\cdot$, whose $\tau^s$-translation is $I\limp
\bot$. The existence of an arrow $\bot\to I\limp\bot$ is, by
Lem.~\ref{lem:basic_arr}, equivalent to the existence of $\bot\otimes I\to\bot$;
this is the unit arrow $\rho$. The induction case involves the sequent $\cdot\to
T_1,\ldots,T_l$, with each $T_i$ hollow; the required arrow exists by
composing the arrows 
given by the induction hypothesis with
$\bot\to\bot\parr\cdots\parr\bot$. The multiset case then follows easily
by considering the cases where $\Vc$ is empty and non-empty.
\end{proof}

\begin{lemma}\label{lem:merg}
Given a multiset $\Tc\in\Tc_1\merge\Tc_2$ of sequents and formulae, there is an
arrow $\tau^s(\Tc_1)\parr\tau^s(\Tc_2)\to\tau^s(\Tc)$ in the free FILL-category.
\end{lemma}
\begin{proof}
We prove this for a single sequent first, by induction on its size. The base
case requires an arrow
$
(\Gamma_1\limp\Delta_1)\parr(\Gamma_2\limp\Delta_2)
\;\to\;
\Gamma_1\otimes\Gamma_2\limp\Delta_1\parr\Delta_2
$
(ref. Lem.~\ref{lm:dist}), which exists by Lem.~\ref{lem:basic_arr}(ii) and (i).
The induction case
follows similarly. The multiset case then follows easily by
considering the cases where $\Tc$ is empty and non-empty.
\end{proof}

\begin{lemma}\label{lem:merg_tens}
Take $X[~]\in X_1[~]\merge X_2[~]$ and $\Tc\in\Tc_1\merge\Tc_2$. Then the
following arrows exist in the free FILL-category for all $A, B,
\Gamma_1 $ and $\Gamma_2$:
\begin{enumerate}
\item $\tau^s(X_1[\Gamma_1\seq A,\Tc_1])\otimes\tau^s(X_2[\Gamma_2\seq B,\Tc_2])
  \;\to\; \tau^s(X[\Gamma_1,\Gamma_2\seq A\otimes B,\Tc])$;
\item $\tau^s(X_1[\Gamma_1\seq A,\Tc_1])\otimes\tau^s(X_2[\Gamma_2,B\seq \Tc_2])
  \;\to\; \tau^s(X[\Gamma_1,\Gamma_2,A\limp B\seq \Tc])$;
\item $\tau^s(X_1[\Gamma_1,A\seq \Tc_1])\otimes\tau^s(X_2[\Gamma_2,B\seq \Tc_2])
  \;\to\; \tau^s(X[\Gamma_1,\Gamma_2,A\parr B\seq \Tc])$;
\end{enumerate}
\end{lemma}
\begin{proof}
All three cases follow by induction on the size of $X[~]$. In all 
three cases
the induction step
is easy, and so we focus on the base cases. By Lem.~\ref{lem:basic_arr} the base
case for (i) requires an arrow:
\begin{equation}\label{eq:mt_1}
(\Gamma_1\limp A\parr\Tc_1)\otimes(\Gamma_2\limp B\parr\Tc_2)\otimes
\Gamma_1\otimes\Gamma_2 \quad\to\quad (A\otimes B)\parr\Tc.
\end{equation}
By the `evaluation' arrows $\varepsilon$ there is an arrow from the left hand
side of \eqref{eq:mt_1} to $(A\parr\Tc_1)\otimes(B\parr\Tc_2)$.
Composing this with weak distributivity takes us to $((A\parr\Tc_1)
\otimes B)\parr\Tc_2$, and then to $(A\otimes B)\parr\Tc_1\parr
\Tc_2$. Lem.~\ref{lem:merg} completes the result. The base cases for (ii) and
(iii) follow by similar arguments (App.~\ref{app:conserv}).
\end{proof}

\begin{theorem}\label{thm:filldn_sound}
For every rule of $\filldn$, if the premises are FILL-valid then so
is the conclusion.
\end{theorem}
\begin{proof}
As FILL-sequents nest no sequents to the left of $\seq$, we can
modify the rules of Fig.~\ref{fig:fbilld} to replace the multisets $\Sc,\Sc'$
of sequents
and formulae with multisets $\Gamma,\Gamma'$ of formulae only, and
remove the hollow multisets of sequents $\Uc$ entirely (see
App.~\ref{app:conserv}).

Therefore by Lem.~\ref{lem:ctxt} the soundness of $pl_1$ amounts to the
existence in the free FILL-category of an arrow
$$
\Gamma\limp(A\otimes\Gamma'\limp\Tc')\parr\Tc \quad\to\quad
\Gamma\otimes A\limp(\Gamma'\limp\Tc')\parr\Tc.
$$
This follows by two uses of Lem.~\ref{lem:basic_arr}(i). Similarly $pr_2$
requires an arrow
$$
\Gamma\limp\Tc\parr A\parr(\Gamma'\limp\Tc') \quad\to\quad
\Gamma\limp\Tc\parr(\Gamma'\limp\Tc'\parr A)
$$
which exists by Lem.~\ref{lem:basic_arr}(ii).

$id^d$: by induction on the size of $X[~]$. The base case requires an arrow
$I\to p\limp p\parr\Vc$, which exists by Lems.~\ref{lem:hollow}
and~\ref{lem:basic_arr}. Induction involves a sequent $\cdot\seq
X[p\seq p,\Vc],\Tc'$, with $\Tc'$ hollow, and hence requires an arrow $I\to I\limp
X[p\seq p,\Vc]\parr\Tc'$. By Lem.~\ref{lem:basic_arr} and the arrow
$I\otimes I\to I$
we need an arrow $I\to X[p\seq p,\Vc]\parr\Tc'$; by the
induction hypothesis we have $I\to X[p\seq p,\Vc]$; this extends to $I\to
X[p\seq p,\Vc]\parr\bot$; Lem.~\ref{lem:hollow} completes the proof.

$\bot^d_l$: by another induction on $X[~]$. The base case $I\to\bot\limp\Vc$
follows by Lems.~\ref{lem:basic_arr} and~\ref{lem:hollow}; induction
follows as with $id^d$.

$\bot^d_r$: By Lem.~\ref{lem:ctxt} and the unit property of $\bot$.

$I^d_l$: By Lem.~\ref{lem:ctxt} we need an arrow $(\Gamma\limp\Tc)
\otimes\Gamma\otimes I\to\Tc$; this exists by the unit property of $I$
and the `evaluation' arrow $\varepsilon$.

$I^d_r$: another induction on $X[~]$. The base case arrow $I\to I\limp I\parr
\Vc$ exists by Lems.~\ref{lem:basic_arr} and~\ref{lem:hollow};
induction follows as with $id^d$.

$\otimes^d_l$, $\limp^d_r$, and $\parr^d_r$ are trivial by the formula
translation.

$\otimes^d_r$: compose the arrow $I\to I\otimes I$ with the arrows defined by
the validity of the premises, then use Lem.~\ref{lem:merg_tens}(i). $\limp^d_l$
and $\parr^d_r$ follow similarly via Lem.~\ref{lem:merg_tens}(ii) and (iii).
\end{proof}

\begin{theorem}
\label{thm:filldn-conserv}
A FILL-formula is $\fill$-valid iff
it is $\filldn$-provable, and BiILL is conservative over FILL.
\end{theorem}
\begin{proof}
By Cors.~\ref{cor:filldc_complete} and~\ref{cor:dc=dn} and
Thms.~\ref{thm:separation} and~\ref{thm:filldn_sound}.
\end{proof}

Note that it is also possible to prove soundness of $\filldn$ w.r.t. FILL syntactically,
i.e., via a translation into Schellinx's sequent calculus for FILL~\cite{Schellinx:Some}. 
See App. G
for details. 

Thm.~\ref{thm:filldn-conserv} gives us a sound and complete calculus 
for $\fill$ that enjoys a genuine subformula property. This in turn allows one to prove
NP-completeness of the tautology problem for $\fill$ (i.e., deciding whether a formula
is provable or not), as we show next. The complexity does not in fact change even
when one adds exclusion to $\fill.$



\begin{theorem}
The tautology problems for BiILL and FILL are NP-complete. 
\end{theorem}
\begin{proof}
({\em Outline.}) 
Membership in NP is proved by showing that every cut-free proof of a formula $A$ in $\biilldn$
can be checked in PTIME in the size of $A$. This is not difficult
to prove given that each connective in $A$ is introduced exactly once in the proof. 
NP-hardness is proved by encoding Constants-Only MLL (COMLL), which is 
NP-hard~\cite{Lincoln94}, in $\filldn$. See App. F for details.
\end{proof}



\section{Conclusion}

We have given three cut-free sequent calculi for FILL without complex 
annotations, showing that, far from being a
curiosity that demands new approaches to proof theory, 
FILL is in a broad family of linear and substructural logics 
captured by display calculi.

Various substructural logics can be defined by using a (possibly
non-associative or non-commutative) multiplicative conjunction and its
left and right residual(s) (implications). Many of these logics have
cut-free sequent calculi with
comma-separated structures in the antecedent and a single formula in
the succedent.
Each of these logics has a dual  logic with disjunction and its
residual(s) (exclusions); their proof theory requires sequents built out
of comma-separated structures in the succedent and a single formula in the
antecedent. 
These logics can then be combined using numerous ``distribution
principles''
\cite{Grishin:Generalisation,DBLP:journals/jphil/Moortgat09}, of which
weak distributivity is but one example.
However, obtaining an adequate sequent calculus for these
combinations is often non-trivial. On the other hand, display calculi
for these logics, their duals, and their combinations, are extremely
easy to obtain using the known methodology for building display
calculi~\cite{Belnap82JPL,Gore:SLD}. We followed this methodology to obtain
$\biill$ in this paper, but needed a conservativity result to ensure the
resulting calculus $\biilldc$ was sound for FILL.
We finally note some specific variations on FILL deserving particular
attention.

\noindent
\textbf{Grishin (a).} Adding the converse of Grishin (b) to FILL
recovers MLL. For example $(B\limp\bot)\parr C\turn B\limp C$ is provable
using Grn(b), but its converse requires Grn(a). Thus there is another
`full' non-classical extension of MILL with Grishin (a) as its
interaction principle \emph{instead} of (b). We do not know what
significance this logic may have.

\noindent
\textbf{Mix rules.} It is easy to give structural rules for the \emph{mix}
sequents $A,B\turn A,B$ and $\Phi\turn\Phi$ which have been studied
in FILL \cite{Cockett:Proof,Bellin:Subnets} and so it is natural
to ask if the results of this paper
can be extended to
them. Intriguingly, our new structural connectives suggest a new mix rule with
sequent form $A<B\turn B>A$ which, given Grishin (b), is stronger than
the mix rule for comma (given Grishin (a), it is weaker).

\noindent
\textbf{Exponentials.} Adding exponentials~\cite{Brauner:Cut} to
our display calculus for
FILL may be possible~\cite{belnap-linear}.

\noindent
\textbf{Additives.} While it has been suggested that FILL could be
extended with additives, the only attempt in the literature is
erroneous~\cite{DBLP:conf/obpdc/GalmicheB95}. It is not clear how easy
this extension would be
\cite[Sec. 1]{Chang:Judgmental}; it is certainly not straightforward
with the display calculus. The problem is most easily seen through the
categorical semantics: additive conjunction $\adcn$ and its
unit $\adtp$ are limits, and $p\parr\mbox{-}$ is a right
adjoint in BiILL but is not necessarily so in FILL. But right adjoints
preserve limits. Then BiILL plus additives is
not conservative over FILL plus additives, because the sequents
$(p\parr q)\adcn(p\parr r)\turn p,(q\land r)$ and $\adtp\turn p,
\adtp$ are valid in the former but not the latter, despite the
absence of $\excl$ or $<$. 
We are
currently investigating solutions.



\bibliographystyle{abbrv}
\bibliography{biblio}

\newpage
\appendix
\section{Display Calculus}\label{app:display}


We outline the conditions that are easily checked to confirm that the
display calculi enjoy cut-elimination (Thm. \ref{thm:cutad}):

\begin{definition}[Belnap's Conditions C1-C8]
The set of display conditions appears in various guises in the literature. Here we follow
the presentation given in Kracht~\cite{kracht-power}.
\begin{description}
\item[(C1)] Each formula variable occurring in some premise of a rule $\rho$ is a subformula
of some formula in the conclusion of~$\rho$.
\item[(C2)] \textit{Congruent parameters} is a relation between parameters of the identical structure variable occurring 
in the premise and conclusion sequents. 
\item[(C3)] Each parameter is congruent to at most one structure variable in the conclusion. Equivalently,
no two structure variables in the conclusion are congruent to each other.
\item[(C4)] Congruent parameters are either all antecedent or all succedent parts of their respective sequent.
\item[(C5)] A formula in the conclusion of a rule~$\rho$ is either the entire
antecedent or the entire succedent. Such a formula is called a \textbf{principal formula} of~$\rho$.
\item[(C6/7)] Each rule is closed under simultaneous substitution of arbitrary structures for congruent
parameters.
\item[(C8)] If there are rules $\rho$ and $\sigma$ with respective
conclusions $X\turn A$ and $A\turn Y$ with formula $A$ principal in
both inferences (in the sense of C5) and if $cut$ is applied to
yield $X\turn Y$, then either $X\turn Y$ is identical to either
$X\turn A$ or $A\turn Y$; or it is possible to pass from the
premises of $\rho$ and $\sigma$ to $X\turn Y$ by means of inferences
falling under $cut$ where the cut-formula always is a proper
subformula of $A$. 
%
\end{description}
\end{definition}

\section{Conservativity of BiILL over FILL}\label{app:conserv}

Fig.~\ref{fig:filld} explicitly gives the proof rules for $\filldn$, the nested
sequent calculus with deep inference for FILL. These are easily derived from
$\biilldn$ (Fig.~\ref{fig:fbilld}).

\begin{figure}[t]
{\small


Propagation rules:

$$
\infer[pl_1]
{X[\Gamma, A \seq (\Gamma' \seq \Tc'), \Tc]}
{
 X[\Gamma \seq (A, \Gamma' \seq \Tc'), \Tc]
}
\qquad
\infer[pr_2]
{X[\Gamma \seq \Tc, (\Gamma' \seq \Tc', A)]}
{X[\Gamma \seq \Tc, A, (\Gamma' \seq \Tc')]}
$$

Identity and logical rules: In branching rules, $X[~] \in X_1[~] \merge X_2[~]$
and $\Tc \in \Tc_1 \merge \Tc_2$.
$$
\infer[id^d]
{X[p \seq p, \Vc]}
{\mbox{$X[~]$ and $\Vc$ are hollow.}}
\qquad
\infer[\punit_l^d]
{X[\punit \seq \Vc]}
{\mbox{$X[~]$ and $\Vc$ are hollow. }}
\qquad
\infer[\punit_r^d]
{X[\Gamma \seq \Tc, \punit]}
{X[\Gamma \seq \Tc]}
$$
$$
\infer[\tunit_l^d]
{X[\Gamma, \tunit \seq \Tc]}
{X[\Gamma \seq \Tc]}
\qquad
\infer[\tunit_r^d]
{X[\cdot \seq \tunit, \Vc]}
{\mbox{$X[~]$ and $\Vc$ are hollow. }}
$$
$$
\infer[\ltens_l^d]
{X[\Gamma, A \ltens B \seq \Tc]}
{X[\Gamma, A, B \seq \Tc]}
\qquad
\infer[\ltens_r^d]
{X[\Gamma_1,\Gamma_2 \seq A \ltens B, \Tc]}
{X_1[\Gamma_1 \seq A, \Tc_1] & X_2[\Gamma_2 \seq B, \Tc_2]}
$$
$$
\infer[\limp_l^d]
{X[\Gamma_1,\Gamma_2, A \limp B \seq \Tc]}
{
X_1[\Gamma_1 \seq A, \Tc_1] & X_2[\Gamma_2, B \seq \Tc_2]
}
\qquad
\infer[\limp_r^d]
{X[\Gamma \seq \Tc, A \limp B]}
{X[\Gamma \seq \Tc, (A \seq B)]}
$$
$$
\infer[\parr_l^d]
{X[\Gamma_1,\Gamma_2, A \parr B \seq \Tc]}
{
 X_1[\Gamma_1, A \seq \Tc_1]
 &
 X_2[\Gamma_2, B \seq \Tc_2]
}
\qquad
\infer[\parr_r^d]
{X[\Gamma \seq A \parr B, \Tc]}
{X[\Gamma \seq A, B, \Tc]}
$$
}
\caption{The deep inference system $\filldn$. }
\label{fig:filld}
\end{figure}

\begin{proof}[Proof of Lemma~\ref{lem:basic_arr}]
This is basic category theory; we give one example to illustrate the techniques
used. Given an arrow $f:A\otimes B\to C$, we get a new arrow $A\to B\limp C$ by
composing $B\limp f$ with the `co-evaluation' arrow $\eta:A\to B\limp(A\otimes
B)$.
\end{proof}

\begin{proof}[Proof of Lemma~\ref{lem:ctxt}]
By induction on the size of $X[~]$. The base case, where $X[~]$ is a hole, is
trivial. The induction case involves a context $\Gamma\seq X[~],\Tc$ and hence
requires an arrow
$$
\Gamma\limp X[S]\parr\Tc \quad\to\quad \Gamma\limp X[T]\parr\Tc.
$$
This exists by the induction hypothesis and the inductive definitions of
Lem.~\ref{def:free}. The validity of $X[S]$ then transfers to $X[T]$ via
composition with the arrow $I\to X[S]$.
\end{proof}

\begin{proof}[Proof of Lemma~\ref{lem:merg_tens}(ii) and (iii)]
(ii): The base case requires an arrow
\begin{equation}\label{eq:mt_2}
(\Gamma_1\limp A\parr\Tc_1)\otimes(\Gamma_2\otimes B\limp \Tc_2)\otimes
\Gamma_1\otimes\Gamma_2\otimes(A\limp B) \quad\to\quad \Tc.
\end{equation}
Applying an evaluation to the left of \eqref{eq:mt_2} gives $(A\parr\Tc_1)
\otimes(\Gamma_2\otimes B\limp \Tc_2)\otimes\Gamma_2\otimes(A\limp B)$; weak
distributivity gives $\Tc_1\parr(A\otimes(\Gamma_2\otimes B\limp \Tc_2)\otimes
\Gamma_2\otimes(A\limp B))$; two more evaluations give $\Tc_1\parr\Tc_2$ and
Lem.~\ref{lem:merg} completes the result.

(iii): The base case requires an arrow
\begin{equation}\label{eq:mt_3}
(\Gamma_1\otimes A\limp \Tc_1)\otimes(\Gamma_2\otimes B\limp \Tc_2)\otimes
\Gamma_1\otimes\Gamma_2\otimes(A\parr B) \quad\to\quad \Tc.
\end{equation}
Two applications of weak distributivity map the left of \eqref{eq:mt_3} to
$$
((\Gamma_1\otimes A\limp \Tc_1)\otimes\Gamma_1\otimes A)\parr
((\Gamma_2\otimes B\limp \Tc_2)\otimes\Gamma_2\otimes B).
$$
Two evaluations and Lem.~\ref{lem:merg} complete the result.
\end{proof}

\section{Annotated Sequent Calculi Proofs}

On the next page we present cut-free proofs of the Bierman example
\eqref{eq:Bierman} in the style of the three cut-free annotated sequent calculi
in the literature: that due to Bierman \cite{Bierman:Note}; that due to Bellin
reported in \cite{Bierman:Note}, and that due to Br\"{a}uner and de Paiva
\cite{Brauner:Formulation}. Note that all three proofs contain the same
sequence of proof rules; strip out the annotations and they are MLL proofs
of the sequent. The difference between the calculi lies in the nature of their
annotations, all of which come into play to verify that the final rule
application, of ($\limp R$), is legal. The reader is invited to compare these
proofs to those presented in the paper using display calculus (Fig.
\ref{Fig:Bierman}) and deep inference (Fig. \ref{fig:bierman-deep}).

\newpage
\begin{landscape}
Bierman-style proof; ($\limp R$) is legal because $v$ and $(w\parr x\limp y)
\parr z$ share no free variables.
\begin{prooftree}
\AxiomC{$v:a\turn v:a$}
\AxiomC{$w:b\turn w:b$}
\BinaryInfC{$v\parr w:a\parr b\turn v:a,w:b$}
\AxiomC{$x:c\turn x:c$}
\BinaryInfC{$(v\parr w)\parr x:(a\parr b)\parr c \turn v:a,w:b,x:c$}
\UnaryInfC{$(v\parr w)\parr x:(a\parr b)\parr c \turn v:a,w\parr x:b\parr c$}
\AxiomC{$y:d\turn y:d$}
\BinaryInfC{$(v\parr w)\parr x:(a\parr b)\parr c,w\parr x\limp y:b\parr c\limp d\turn v: a,y:d$}
\AxiomC{$z:e\turn z:e$}
\BinaryInfC{$(v\parr w)\parr x:(a\parr b)\parr c,(w\parr x\limp y)\parr z:(b\parr c\limp d)\parr e\turn v:a,y:d,z:e$}
\UnaryInfC{$(v\parr w)\parr x:(a\parr b)\parr c,(w\parr x\limp y)\parr z:(b\parr c\limp d)\parr e\turn v:a,y\parr z:d\parr e$}
\UnaryInfC{$(v\parr w)\parr x:(a\parr b)\parr c\turn v:a,\lambda(w\parr x\limp y)\parr z^{(b\parr c\limp d)\parr e}.(y\parr z):(b\parr c\limp d)\parr e\limp d\parr e$}
\end{prooftree}

Bellin-style proof; ($\limp R$) is legal because $r$ is not free in $\mbox{let }
t\mbox{ be }u\parr\mbox{- in let }u\mbox{ be }v\parr\mbox{- in }v$. We apologise
for the extremely small font size necessary to fit this proof on the page.
{\tiny
\begin{prooftree}
\AxiomC{$v:a\turn v:a$}
\AxiomC{$w:b\turn w:b$}
\BinaryInfC{$u:a\parr b\turn \mbox{let }u\mbox{ be }v\parr\mbox{- in }v:a,\mbox{let }u\mbox{ be -}\parr w\mbox{ in }w:b$}
\AxiomC{$x:c\turn x:c$}
\BinaryInfC{$t:(a\parr b)\parr c \turn \mbox{let }t\mbox{ be }u\parr\mbox{- in let }u\mbox{ be }v\parr\mbox{- in }v:a,\mbox{let }t\mbox{ be }u\parr\mbox{- in let }u\mbox{ be -}\parr w\mbox{ in }w:b,\mbox{let }t\mbox{ be -}\parr x\mbox{ in }x:c$}
\UnaryInfC{$t:(a\parr b)\parr c \turn \mbox{let }t\mbox{ be }u\parr\mbox{- in let }u\mbox{ be }v\parr\mbox{- in }v:a,(\mbox{let }t\mbox{ be }u\parr\mbox{- in let }u\mbox{ be -}\parr w\mbox{ in }w)\parr(\mbox{let }t\mbox{ be -}\parr x\mbox{ in }x):b\parr c$}
\AxiomC{$y:d\turn y:d$}
\BinaryInfC{$t:(a\parr b)\parr c,s:b\parr c\limp d\turn  \mbox{let }t\mbox{ be }u\parr\mbox{- in let }u\mbox{ be }v\parr\mbox{- in }v:a,(s(\mbox{let }t\mbox{ be }u\parr\mbox{- in let }u\mbox{ be -}\parr w\mbox{ in }w)\parr(\mbox{let }t\mbox{ be -}\parr x\mbox{ in }x)):d$}
\AxiomC{$z:e\turn z:e$}
\BinaryInfC{$t:(a\parr b)\parr c,r:(b\parr c\limp d)\parr e\turn  \mbox{let }t\mbox{ be }u\parr\mbox{- in let }u\mbox{ be }v\parr\mbox{- in }v:a,\mbox{let } r\mbox{ be }s\parr\mbox{- in }(s(\mbox{let }t\mbox{ be }u\parr\mbox{- in let }u\mbox{ be -}\parr w\mbox{ in }w)\parr(\mbox{let }t\mbox{ be -}\parr x\mbox{ in }x)):d,\mbox{let }s\mbox{ be -}\parr z\mbox{ in }z:e$}
\UnaryInfC{$t:(a\parr b)\parr c,r:(b\parr c\limp d)\parr e\turn  \mbox{let }t\mbox{ be }u\parr\mbox{- in let }u\mbox{ be }v\parr\mbox{- in }v:a,(\mbox{let } r\mbox{ be }s\parr\mbox{- in }(s(\mbox{let }t\mbox{ be }u\parr\mbox{- in let }u\mbox{ be -}\parr w\mbox{ in }w)\parr(\mbox{let }t\mbox{ be -}\parr x\mbox{ in }x)))\parr(\mbox{let }s\mbox{ be -}\parr z\mbox{ in }z):d\parr e$}
\UnaryInfC{$t:(a\parr b)\parr c\turn \mbox{let }t\mbox{ be }u\parr\mbox{- in let }u\mbox{ be }v\parr\mbox{- in }v:a,\lambda r^{(b\parr c\limp d)\parr e}.(\mbox{let } r\mbox{ be }s\parr\mbox{- in }(s(\mbox{let }t\mbox{ be }u\parr\mbox{- in let }u\mbox{ be -}\parr w\mbox{ in }w)\parr(\mbox{let }t\mbox{ be -}\parr x\mbox{ in }x)))\parr(\mbox{let }s\mbox{ be -}\parr z\mbox{ in }z):(b\parr c\limp d)\parr e\limp d\parr e$}
\end{prooftree}
}

Br\"{a}uner and de Paiva-style proof; ($\limp R$) is legal because $(b\parr c
\limp d)\parr e$ is not related to $a$.
\begin{prooftree}
\AxiomC{ \ }\LeftLabel{$(a,a)$\;}
\UnaryInfC{$a\turn a$}
\AxiomC{ \ }\RightLabel{\;$(b,b)$}
\UnaryInfC{$b\turn b$}\LeftLabel{$(a\parr b,a),(a\parr b,b)$\;}
\BinaryInfC{$a\parr b\turn a,b$}
\AxiomC{ \ }\RightLabel{\;$(c,c)$}
\UnaryInfC{$c\turn c$}\LeftLabel{$((a\parr b)\parr c,a),((a\parr b)\parr c,b),((a\parr b)\parr c,c)$\;}
\BinaryInfC{$(a\parr b)\parr c \turn a,b,c$}\LeftLabel{$((a\parr b)\parr c,a),((a\parr b)\parr c,b\parr c)$\;}
\UnaryInfC{$(a\parr b)\parr c \turn a,b\parr c$}
\AxiomC{ \ }\RightLabel{$(d,d)$}
\UnaryInfC{$d\turn d$}\LeftLabel{$((a\parr b)\parr c,a),((a\parr b)\parr c,d),(b\parr c\limp d,d)$\;}
\BinaryInfC{$(a\parr b)\parr c,b\parr c\limp d\turn a,d$}
\AxiomC{ \ }\RightLabel{\;$(e,e)$}
\UnaryInfC{$e\turn e$}\LeftLabel{$((a\parr b)\parr c,a),((a\parr b)\parr c,d),((b\parr c\limp d)\parr e,d),((b\parr c\limp d)\parr e,e)$\;}
\BinaryInfC{$(a\parr b)\parr c,(b\parr c\limp d)\parr e\turn a,d,e$}\LeftLabel{$((a\parr b)\parr c,a),((a\parr b)\parr c,d\parr e),((b\parr c\limp d)\parr e,d\parr e)$\;}
\UnaryInfC{$(a\parr b)\parr c,(b\parr c\limp d)\parr e\turn a,d\parr e$}\LeftLabel{$((a\parr b)\parr c,a),((a\parr b)\parr c,(b\parr c\limp d)\parr e\limp d\parr e)$\;}
\UnaryInfC{$(a\parr b)\parr c\turn a,(b\parr c\limp d)\parr e\limp d\parr e$}
\end{prooftree}

\end{landscape}

\section{The Shallow Nested Sequent Calculus}
\label{apx:shallow}

A structure can be interpreted as a multiset of nested sequents by 
replacing both $>$ and $<$ with the sequent arrow $\seq$, and interpreting the structural
connective `,' (comma) as multiset union, and $\Phi$ as the empty multiset. 
That is, the structure of a nested sequent incorporates implicitly the associativity and
commutativity of comma, and its unit, via the multiset structure. Conversely, a
nested sequent can be translated to an equivalence class of structures (modulo 
the associativity, commutativity and unit laws for `,') 
by replacing sequent arrows in negative positions with $<$, and those in positive positions with $>$. 
Given a nested sequent $X$, we shall write $\seqtodisp{X}$ to denote the corresponding (equivalence class of) 
structure in display calculus. Conversely, give a structure $X$, we write 
$\disptoseq{X}$ to denote the multiset of formulas/sequents that correspond to $X.$
\\

\noindent
{\bf Theorem \ref{thm:biill-sn-eq-dc}.} 
A formula $B$ is cut-free provable in $\biillsn$ iff it is cut-free provable in $\biilldc$. 

\begin{proof}
We show that cut-free $\biillsn$ can simulate cut-free $\biilldc$ and vice versa. To prove this, we need
to generalise slightly the statement to the following:
\begin{itemize}
\item If $(X \vdash Y)$ is cut-free provable in $\biilldc$ then $(\disptoseq{X} \seq \disptoseq{Y})$
is cut-free provable in $\biillsn$.
\item If $(X \seq Y)$ is cut-free provable in $\biillsn$ then $\seqtodisp{X} \vdash \seqtodisp{Y}$
is cut-free provable in $\biilldc.$
\end{itemize}
The first statement is easy, since the rules of $\biillsn$ are more general than $\biilldc.$ 
We show here the other direction. We illustrate here the case for the $\parr_l$ rule. 
For simplicity, we omit applications of structural
rules for associativity, commutativity and unit, and obvious applications of display postulates. 
$$
\infer[\parr_l]
{\Sc, \Sc', A \parr B \seq \Tc, \Tc'}
{\Sc, A \seq \Tc & \Sc', B \seq \Tc'}
\qquad
\leadsto
\qquad
\infer[]
{\seqtodisp{\Sc}, \seqtodisp{\Sc'}, A \parr B \vdash \seqtodisp{\Tc}, \seqtodisp{\Tc'}}
{
\infer[]
{\seqtodisp{\Sc},  A \parr B \vdash \seqtodisp{\Sc'} > (\seqtodisp{\Tc}, \seqtodisp{\Tc'})}
{
 \infer[]
 {\seqtodisp{\Sc},  A \parr B \vdash (\seqtodisp{\Sc'} > \seqtodisp{\Tc'}), \seqtodisp{\Tc}}
 {
  \infer[]
  {  A \parr B \vdash \seqtodisp{\Sc} > ((\seqtodisp{\Sc'} > \seqtodisp{\Tc'}), \seqtodisp{\Tc})}
  {
   \infer[]
   { A \parr B \vdash (\seqtodisp{\Sc} > \seqtodisp{\Tc}), (\seqtodisp{\Sc'} > \seqtodisp{\Tc'})}
   {
    \infer[]
    {A \vdash (\seqtodisp{\Sc} > \seqtodisp{\Tc})}
    {\seqtodisp{\Sc}, A \vdash \seqtodisp{\Tc}}
    &
    \infer[]
    {B \vdash (\seqtodisp{\Sc'} > \seqtodisp{\Tc'})}
    {\seqtodisp{\Sc'}, B \vdash \seqtodisp{\Tc'}}
   }
  }
 }
}
}
$$
\end{proof}

\section{The Equivalence Between Shallow and Deep Inference Calculi}
\label{apx:equiv}

\subsection{From deep inference to shallow inference}
\label{apx:deep-to-shallow}

\begin{lemma}[Weakening of hollow sequents]
\label{lm:weak}
The following rules are cut-free derivable in $\biillsn$:
$$
\infer[w_l]
{X, \Uc \seq \Vc}
{\Uc \seq \Vc}
\qquad
\infer[w_r]
{\Uc \seq \Vc, X}
{\Uc \seq \Vc}
$$
provided $X$ is a hollow sequent.
\end{lemma}
\begin{proof}
By induction on the size of $X.$
\end{proof}


\begin{lemma}
\label{lm:idd}
The rule $id^d$ is cut-free derivable in $\biillsn$.
\end{lemma}
\begin{proof}
We show that $X[\Sc, A \seq A, \Tc]$ is provable in $\biillsn$,
where $X[~]$, $\Sc$ and $\Tc$ are hollow. 
We show the case where $X[~]$ is a positive context; the other case where $X[~]$ is negative
can be proved dually. 
Note that by the display property (Proposition \ref{prop:display}), the sequent $X[\Sc, A \seq A, \Tc]$
is display-equivalent to $\Uc, \Sc, A \seq A, \Tc$ for some $\Uc$.
Clearly the structure $\Uc$ here must be a multiset of hollow sequents. The derivation
is thus constructed as follows: 
$$
\infer[\mbox{Prop. \ref{prop:display} }]
{X[\Sc, A \seq A, \Tc]}
{
\infer[wr;wl]
{\Uc, \Sc, A \seq A, \Tc}
{
 \infer[id]
 {A \seq A}{}
}
}
$$
\end{proof}

\noindent{\bf{Lemma \ref{lm:dist}.}} 
The following rules are derivable in $\biillsn$ without cut:
$$
\infer[dist_l]
{(\Xc_1, \Xc_2  \seq \Yc_1,\Yc_2), \Uc \seq \Vc}
{
(\Xc_1 \seq \Yc_1), (\Xc_2 \seq \Yc_2), \Uc \seq \Vc
}
\qquad
\infer[dist_r]
{\Uc \seq \Vc, (\Xc_1, \Xc_2  \seq \Yc_1,\Yc_2)}
{
 \Uc \seq \Vc, (\Xc_1 \seq \Yc_1), (\Xc_2 \seq \Yc_2)
}
$$
\begin{proof}
We show here a derivation of $dist_l$. The $dist_r$ rule can be derived similarly. 
$$
\infer[rp_2]
{(\Xc_1, \Xc_2  \seq \Yc_1,\Yc_2), \Uc \seq \Vc}
{
 \infer[drp_1]
 {(\Xc_1, \Xc_2  \seq \Yc_1,\Yc_2) \seq (\Uc \seq \Vc)}
 {
   \infer[drp_2]
   {\Xc_1, \Xc_2  \seq \Yc_1,\Yc_2, (\Uc \seq \Vc)}
   {
    \infer[gl]
    {(\Xc_1, \Xc_2  \seq \Yc_2) \seq \Yc_1, (\Uc \seq \Vc)}
    {
      \infer[rp_2]
      {\Xc_1, (\Xc_2  \seq \Yc_2) \seq \Yc_1, (\Uc \seq \Vc)}  
      {
       \infer[gr]
       {\Xc_1 \seq ((\Xc_2  \seq \Yc_2) \seq \Yc_1, (\Uc \seq \Vc))}  
       {
         \infer[drp_2]
         {\Xc_1 \seq \Yc_1, ((\Xc_2  \seq \Yc_2) \seq (\Uc \seq \Vc))}  
         {
           \infer[rp_1]
           {(\Xc_1 \seq \Yc_1) \seq ((\Xc_2  \seq \Yc_2) \seq (\Uc \seq \Vc))}  
           {
            \infer[rp_1]
            {(\Xc_1 \seq \Yc_1), (\Xc_2  \seq \Yc_2) \seq (\Uc \seq \Vc)}
            {
              {(\Xc_1 \seq \Yc_1), (\Xc_2  \seq \Yc_2), \Uc \seq \Vc}
            }  
           }
         }
      }
     }
    }
    }
 }
}
$$
\end{proof}

\subsection{From shallow inference to deep inference}
\label{apx:shallow-to-deep}

\begin{lemma}
\label{lm:idd-permute}
Suppose the $id^d$ rule is applicable to $X$. Suppose also that $X$ is the premise of
an instance of a rule in $\{rp_1,rp_2,drp_1,drp_2,gl,gr \}$ and suppose $X'$
is the conclusion of the same rule instance. Then $X'$ is derivable in $\biilldn.$
\end{lemma}
\begin{proof}
Since $id^d$ is applicable to $X$, it must be the case that
$X = Y[\Sc, A \seq A, \Tc]$ for some $A$, hollow context $Y[~]$, 
and hollow sequents $\Sc$ and $\Tc.$
We do case analyses on how the rule $\rho$ affects $X$. If $\rho$ changes the
structure of $Y[~]$ only, but leave $[\Sc, A \seq A, \Tc]$ intact, 
i.e., $X' = Y'[\Sc, A \seq A, \Tc]$, then obviously $Y'[~]$ must also be
a hollow sequent, so the $id^d$ rule is applicable. 
The interesting case is when $\rho$ affects the subsequent $(\Sc, A \seq A, \Tc)$, i.e.,
when exactly of the $A$'s is moved by $\rho$ to a different nested sequent. 
We show here the interesting cases; the
others can be proved similarly. 
In all cases, these structural rules
can be replaced by propagation rules of $\biilldn.$
\begin{itemize}
\item $\rho$ is
$$
\infer[drp_1]
{(\Sc, A \seq \Uc) \seq A, \Vc}
{
\Sc, A \seq A, \Uc, \Vc
}
$$
where $\Tc = (\Uc, \Vc).$ Then the derivation of $X'$ is as follows:
$$
\infer[pl_2]
{(\Sc, A \seq \Uc) \seq A, \Vc}
{
\infer[id^d]
{A, (\Sc \seq \Uc) \seq A, \Vc}
{}
}
$$

\item $\rho$ is
$$
\infer[rp_1]
{\Uc, A \seq (\Vc \seq A, \Tc)}
{
\Uc, \Vc, A \seq A, \Tc
}
$$
where $\Sc = (\Uc, \Vc).$ The sequent $X'$ is derived as follows:
$$
\infer[pr_2]
{\Uc, A \seq (\Vc \seq A, \Tc)}
{
\infer[id^d]
{\Uc, A \seq (\Vc \seq \Tc), A}
{}
}
$$
\end{itemize}
\end{proof}

\begin{lemma}
\label{lm:const}
Suppose the $\punit$ rule (resp. the $\tunit$ rule)is applicable to $X.$ Suppose $X$ is
the premise of an instance of a rule in $\rho \in \{rp_1,rp_2,drp_1,drp_2,gl,gr \}$ and suppose
$X'$ is the conclusion of the same rule. Then $X'$ is derivable in $\biilldn.$
\end{lemma}

To prove the following lemma, it is useful to consider a generalisation of the rules $gl$, $gr$, $drp_2$ and $rp_2$: 
$$
\infer[eg]
{\Sc, \Sc' \seq \Tc, \Tc'}
{(\Sc \seq \Tc), \Sc' \seq \Tc'}
\qquad
\infer[ig]
{\Sc, \Sc' \seq \Tc, \Tc'}
{
\Sc \seq (\Sc' \seq \Tc'), \Tc
}
$$
These two rules can be derived using $gl$, $gr$, $drp_2$ and $rp_2$ as follows:
$$
\infer[drp_2]
{\Sc, \Sc' \seq \Tc, \Tc'}
{
\infer[gl]
{(\Sc, \Sc' \seq \Tc) \seq \Tc'}
{
\Sc', (\Sc \seq \Tc) \seq \Tc'
}
}
\qquad
\qquad
\infer[rp_2]
{\Sc, \Sc' \seq \Tc, \Tc'}
{
\infer[gr]
{\Sc \seq (\Sc' \seq \Tc, \Tc')}
{
\Sc \seq (\Sc' \seq \Tc'), \Tc
}
}
$$
Conversely, $gl$, $gr$, $drp_2$ and $rp_2$ can be derived using $drp_1,rp_1,eg$ and $ig$:
$$
\infer[gl]
{(\Sc, \Tc \seq \Sc') \seq \Tc'}
{(\Sc \seq \Sc'), \Tc \seq \Tc'}
\qquad
\leadsto
\qquad
\infer[drp_1]
{(\Sc, \Tc \seq \Sc') \seq \Tc'}
{
 \infer[eg]
 {\Sc, \Tc \seq \Sc', \Tc'}
 {
  (\Sc \seq \Sc'), \Tc \seq \Tc'
 }
}
$$
$$
\infer[gr]
{\Sc \seq (\Sc' \seq \Tc', \Tc)}
{
 \Sc \seq (\Sc' \seq \Tc'), \Tc
}
\qquad
\leadsto
\qquad
\infer[rp_1]
{\Sc \seq (\Sc' \seq \Tc', \Tc)}
{
 \infer[ig]
 {\Sc, \Sc' \seq \Tc', \Tc}
 {
  \Sc \seq (\Sc' \seq \Tc'), \Tc
 }
}
$$
Note that $drp_2$ and $rp_2$ are just special cases of $eg$ and $ig.$
Lemma~\ref{lm:permute} then follows from the following lemma. 

\begin{lemma}
The rules $drp_1,rp_1,eg$ and $ig$ permute up over all logical rules of $\biilldn.$
\end{lemma}
\begin{proof}
In the following, we omit trivial cases where the structural rule being applied does not affect
the (sub)sequent where the principal formula of the logical rule resides. 

For permutation over the propagation rules, the non-trivial cases are those where the structural rule enables
the propagation to happen. We look at some non-trivial cases here; the others are similar. In all cases,
the propagation may need to be replaced by one or more propagation rules, or may be absorbed 
by the structural rule. 
\begin{itemize}
\item $pl_1$ over $drp_1$: 
$$
\infer[drp_1]
{(\Sc, B \seq \Tc) \seq (\Uc \seq \Vc), \Tc'}
{
 \infer[pl_1]
 {\Sc, B \seq \Tc, (\Uc \seq \Vc), \Tc' }
 {
  \Sc \seq \Tc, (B, \Uc \seq \Vc), \Tc, \Tc'
 }
}
\qquad
\leadsto
\qquad
\infer[pl_2]
{(\Sc, B \seq \Tc) \seq (\Uc \seq \Vc), \Tc'}
{
 \infer[pl_1]
 {B, (\Sc \seq \Tc) \seq (\Uc \seq \Vc), \Tc'}
 {
  \infer[drp_1]
  {(\Sc \seq \Tc) \seq (B, \Uc \seq \Vc), \Tc'}
  {
    \Sc \seq \Tc,  (B, \Uc \seq \Vc), \Tc'
  }
 }
}
$$

\item $pl_1$ over $rp_1$:
$$
\infer[rp_1]
{\Sc, B \seq (\Tc \seq (\Uc \seq \Vc), \Tc')}
{
\infer[pl_1]
{\Sc, B, \Tc \seq (\Uc \seq \Vc), \Tc'}
{
 {\Sc, \Tc \seq (B, \Uc \seq \Vc), \Tc'}
}
}
\qquad
\leadsto
\qquad
\infer[pl_1]
{\Sc, B \seq (\Tc \seq (\Uc \seq \Vc), \Tc')}
{
\infer[pl_1]
{\Sc \seq (\Tc, B \seq (\Uc \seq \Vc), \Tc')}
{
\infer[rp_1]
{\Sc \seq (\Tc \seq (B, \Uc \seq \Vc), \Tc')}
{
 {\Sc, \Tc \seq (B, \Uc \seq \Vc), \Tc'}
}
}
}
$$

\item $pl_1$ over $ig$:
$$
\infer[ig]
{\Sc, B, \Sc' \seq \Tc, \Tc'}
{
\infer[pl_1]
{\Sc, B \seq (\Sc' \seq \Tc'), \Tc}
{
\Sc \seq (B, \Sc' \seq \Tc'), \Tc
}
}
\qquad
\leadsto
\qquad
\infer[ig]
{\Sc, B, \Sc' \seq \Tc, \Tc'}
{
\Sc \seq (B, \Sc' \seq \Tc'), \Tc
}
$$

\item The cases for permutation over $pr_1$ can be done similarly; 
replacing left-propagation rules ($pl_1$, $pl_2$)
with right-propagation rules ($pr_1$, $pr_2$).  

\item The cases for permutation over $pl_2$ are mostly straightforward.
The only non-trivial case is the following: 
$$
\infer[eg]
{\Sc, B, \Sc' \seq \Tc, \Tc'}
{
 \infer[pl_2]
 {(\Sc, B \seq \Tc), \Sc' \seq \Tc'}
 {
  B, (\Sc \seq \Tc), \Sc' \seq \Tc'
 } 
}
\qquad
\leadsto
\qquad
\infer[eg]
{\Sc, B, \Sc' \seq \Tc, \Tc'}
{B, (\Sc \seq \Tc), \Sc' \seq \Tc'}
$$
\item The cases for permutation over $pr_2$ can be done similarly to the cases
for $pl_2.$
\end{itemize}

Permutation over non-branching logical rules is trivial, as the sequent structure of the conclusion 
of a logical rule is preserved in the premise. 
For the branching rules, we look at the case with $\limp_l$, which is slightly non-trivial.
The rest can be proved similarly. 

In the following, we show only non-trivial interactions between the structural rules
and $\limp_l$, i.e., those in which the principal formula of $\limp_l$ is moved
by the display rule.

\begin{itemize}
\item $\limp_l$ over $drp_1$: 
\[
\infer[drp_1]
{(\Sc, C \limp B \seq \Tc) \seq \Uc}
{
\infer[\limp_l]
{\Sc, C \limp B \seq \Tc, \Uc}
{
{\Sc_1 \seq C, \Tc_1, \Uc_1}
&
{\Sc_2, B \seq \Tc_2, \Uc_2}
}
}
\qquad
\leadsto
\]
\[
\qquad
\infer[\limp_l]
{(\Sc, C \limp B \seq \Tc) \seq \Uc}
{
\infer[drp_1]
{(\Sc_1 \seq C, \Tc_1) \seq \Uc_1}
{\Sc_1 \seq C, \Tc_1, \Uc_1}
&
\infer[drp_1]
{(\Sc_2, B \seq \Tc_2) \seq \Uc_2}
{\Sc_2, B \seq \Tc_2, \Uc_2}
}
\]

\item $\limp_l$ over $rp_1$:
\[
\infer[rp_1]
{\Sc, C \limp B \seq (\Tc \seq \Uc)}
{
\infer[\limp_l]
{\Sc, C \limp B, \Tc \seq \Uc}
{
 {\Sc_1, \Tc_1 \seq C, \Uc_1}
&
 {\Sc_2, B \seq \Tc_2, \Uc_2}
}
}
\qquad
\leadsto
\]
\[
\qquad
\infer[pl_1]
{\Sc, C \limp B \seq (\Tc \seq \Uc)}
{
 \infer[\limp_l]
 {\Sc \seq (C \limp B, \Tc \seq \Uc)}
 {
  \infer[rp_1] 
  {\Sc_1 \seq (\Tc_1 \seq C, \Uc_1)}
  {\Sc_1, \Tc_1 \seq C, \Uc_1}
  &
  \infer[rp_1]
  {\Sc_2 \seq (\Tc_2, B \seq \Uc_2)}
  {\Sc_2, \Tc_2, B \seq \Uc_2}
 }
}
\]
Notice that we need to use the propagation rule $pl_1$ to push $rp_1$ over $\limp_l$.
The only other case where a propagation rule is used is when permuting 
$drp_1$ over $\excl_r$; in this case the propagation rule needed is $pr_1$. 

\item $\limp_l$ over $eg$:
$$
\infer[eg]
{\Sc, C \limp B, \Sc' \seq \Tc, \Tc'}
{
 \infer[\limp_l]
 {(\Sc, C \limp B \seq \Tc), \Sc' \seq \Tc'}
 {
  {(\Sc_1 \seq C, \Tc_1), \Sc_1' \seq \Tc_1'}
  &
  {(\Sc_2, B \seq \Tc_2) \seq \Tc_2'}
 }
}
$$
$$
\leadsto
\qquad
\infer[\limp_l]
{\Sc, C \limp B, \Sc' \seq \Tc, \Tc'}
{
 \infer[eg]
 {\Sc_1, \Sc_1' \seq C, \Tc_1, \Tc_1'}
 {(\Sc_1 \seq C, \Tc_1), \Sc_1' \seq \Tc_1'}
 &
 \infer[eg]
 {\Sc_2, \Sc_2', B \seq \Tc_2, \Tc_2'}
 {(\Sc_2, B \seq \Tc_2), \Sc_2' \seq \Tc_2'}
}
$$

\item $\limp_l$ over $ig$:
$$
\infer[ig]
{\Sc, \Sc', C \limp B \seq \Tc, \Tc'}
{
 \infer[\limp_l]
 {\Sc \seq (\Sc', C \limp B \seq \Tc'), \Tc}
 {
  {\Sc_1 \seq (\Sc_1' \seq C, \Tc_1'), \Tc_1}
  &
  {\Sc_2 \seq (\Sc_2', B \seq \Tc_2'), \Tc_2}
 }
}
$$
$$
\leadsto
\qquad
\infer[\limp_l]
{\Sc, \Sc', C \limp B \seq \Tc,\Tc'}
{
 \infer[ig]
 {\Sc_1, \Sc_1' \seq C, \Tc_1, \Tc_1'} 
 {\Sc_1 \seq (\Sc_1' \seq C, \Tc_1'), \Tc_1}
 &
 \infer[ig]
 {\Sc_2, \Sc_2', B \seq \Tc_2, \Tc_2'}
 {\Sc_2 \seq (\Sc_2', B \seq \Tc_2'), \Tc_2'}
}
$$
\end{itemize}
\end{proof}

\section{Proof that BiILL and FILL  are NP-complete}
\label{apx:np}

\begin{lemma}
The tautology problem for BiLLL is in NP
\end{lemma}
\begin{proof}
We shall utilise the deep inference system $\biilldn$ to prove this. 
To show membership in NP, it is sufficient to show that every proof of a formula $B$ in $\biilldn$
can be checked in polynomial time. 
So suppose we are given a proof $\Pi$ of a formula $B$ in $\biilldn$.
We first establish show that the size of $\Pi$ is bounded polynomially by $|B|$, and show
that checking validity of each inference step of $\Pi$ is decidable in PTIME in the size of $|B|.$

We assume a representation of formulas as ordered trees, with nodes
labelled with connectives and propositional variables. 
A nested sequent is represented
as an unordered tree of ordered pairs of lists of formulae. 
The edges are labelled with polarity information ($+$ or $-$)
to indicate whether a child node of a parent node is nested to the left or the right of the sequent
represented by the parent node. 
We further assume that each occurrence of $\limp$ and $\excl$ in $B$ is labelled with a unique
identifier. Such a labelling only introduces at most a polynomial overhead of 
the size of $B$, so is inconsequential
as far as proving the upper bound in NP is concerned. Since each sequent arrow is created (reading the rules
upwards) by decomposing exactly one occurrence of $\limp$ or $\excl$, we can assume w.l.o.g. that
each node in a nested sequent is similarly uniquely labelled. Given this, it is easy to see 
that checking whether two trees of ordered pairs of lists 
represent the same nested sequent can be done in PTIME in the size of the trees. An inspection
on the rules of $\biilldn$ shows that the size of each nested sequent in $\Pi$ is bounded
by $|B|$, since each introduction rule replaces one formula connective with
zero or one structural connective. So checking equality between two nested sequents (or contexts) in $\Pi$ 
can be done in PTIME in the size of $B.$ 
We need to further show that, in the branching rules in $\Pi$, that merging of contexts and
sequents happen only along the same labelled nodes. 
So, checking whether the merging relation holds 
between three sequents (or contexts) in a branching rule in $\Pi$
can also be decided in PTIME in the size of $B.$

Notice that every occurrence of 
a  propositional variable or a constant in $B$ appears exactly once in either an $id^d$ rule or a 
constant rule, so the number of leaves in $\Pi$ is bounded by $|B|.$
That means that the number of branches in $\Pi$ is bounded by  $|B|.$ 
In every branch in $\Pi$, the number of logical rules is bounded by $|B|$, because each
connective is introduced exactly once in $\Pi.$
Now we need also to account for the number of propagation rules. Notice that the propagation rules
are non-invertible, and each formula occurrence can be propagated at most $k$ times, where $k$ is
the number of $\limp$ and $\excl$ occurring in $B$. The number of formula occurrences in
a nested sequent in $\Pi$ is bounded by $|B|$, so there can be at most $k \times |B|$ propagation
rules that can be successively applied to a nested sequent. 
The length of each branch in $\Pi$ is bounded by $|B| + (|B|/2) \times k \times |B|$, i.e.,
the number of logical rules, plus the number of propagation rules in between every pair
of logical rules. So the length of each branch is bounded by $O(|B|^3).$
That means that number of nodes in $\Pi$ is bounded by  $O(|B|^4)$. 
Now the size of each sequent in the node is obviously bounded by $B$, so
the total size of $\Pi$ is bounded by $O(|B|^5)$. It remains to show that
checking whether each inference in $\Pi$ is valid is decidable in PTIME in the size of $B$. 
The slightly non-trivial
bit is to decide whether the splitting of the contexts in branching rules are valid, e.g.,
whether $X[~] \in X_1[~] \merge X_2[~]$, etc. As discussed above, this can be done in PTIME as well
given our unique labelling assumption. 
\end{proof}

\begin{lemma}
The tautology problem for $\filldn$ is NP-hard. 
\end{lemma}
\begin{proof}
  We show how to encode constants-only multiplicative
  linear logic (COMLL), which is MLL without any propositional
  variables and is known to be NP-complete~\cite{Lincoln94}. Since
  COMLL has no propositional variables, every COMLL formula in nnf
  has no negation (or implication) in it. So it is
  enough to show that every COMLL formula in nnf is valid iff it is
  provable in $\filldn$. The restriction of $\filldn$ to the COMLL
  connectives gives us exactly COMLL, since without implication
  (negation), the proof system degenerates into the usual classical
  sequent calculus for COMLL. 
\end{proof}

\section{Conservativity via Schellinx's sequent calculus}
\label{apx:conserv-schellinx}

We give here an alternative proof that $\biill$ is conservative over FILL via
the sequent calculus for FILL that was proposed by Schellinx~\cite{Schellinx:Some}.
The rules of this calculus are those for MLL, except 
that the $\limp$-right rule is replaced with: 
$$
\infer[]
{\Gamma \vdash C \limp B}
{\Gamma, C \vdash B}
$$
We show that every derivation of a formula in $\filldn$ can be translated to a derivation
of the same formula in Schellinx's calculus, possibly using cuts.

We shall assume that formulas are equivalent modulo associativity and commutativity
of $\ltens$ and $\parr$. This is not necessary but it helps to simplify the proof. 
Formally this just means that there are implicit cuts in the following constructions that
we omit, i.e., those needed to allow replacement of a formula with its equivalent one. 
Given this convention, we shall omit parentheses when writing a tensor (par) of multiple
formulas, e.g., $(B_1 \ltens B_2 \ltens B_3).$
Given $\Gamma = \{B_1,\dots,B_n\}$, 
we shall write $\ltens(\Gamma)$ to denote  the formula
$(B_1 \ltens \cdots \ltens B_n)$, and similarly, $\parr(\Gamma)$
denotes $(B_1 \parr \cdots \parr B_n).$ 

In the following proofs, we shall be working with formula interpretations of nested
sequents, using the translation functions $\tau^a$ and $\tau^s$ defined
in Definition~\ref{def:tau_nested}.
Strictly speaking, the empty sequent $\sunit \seq \sunit$ would be interpreted (positively) as 
$\tunit \ltens \tunit \limp 0 \parr 0$, because of the way we separate formulas from sequents in the
definition of $\tau$. But when using the translation function $\tau$, 
we shall treat $\tunit \ltens \cdots \ltens \tunit$ as simply $\tunit$, to simplify
presentation. This is harmless as they are logically equivalent (alternatively we could write
a more complicated translation function just to take care of this minor syntactic bureaucracy). 

\begin{lemma}
\label{lm:conserv-merge}
Let $X$ and $Y$ be FILL sequents. If $Z \in X \merge Y$, then
the formula
$$
\tau^s(X) \parr \tau^s(Y) \limp \tau^s(Z)
$$
is provable in FILL.
\end{lemma}
\begin{proof}
By induction on the structure of $Z.$
So suppose $Z$ is 
$$
(\Gamma, \Delta \seq \Gamma', \Delta', Z_1, \dots, Z_n)
$$
and $X$ and $Y$ are, respectively,  
$$
(\Gamma \seq \Gamma', X_1, \dots, X_n) \qquad (\Delta \seq \Delta', Y_1, \dots, Y_n)
$$
where $Z_i \in X_i \merge Y_i$, for every $i \in \{1,\dots,n\}.$
By the induction hypothesis, we have for each $i$: $\tau^s(X_i) \parr \tau^s(Y_i) \limp \tau^s(Z_i).$
To prove $\tau^s(X) \parr \tau^s(Y) \limp \tau^s(Z)$ it is enough to show that the following
sequent is derivable in FILL:
$$
\begin{array}{l}
\ltens(\Gamma) \limp \parr(\Gamma',\tau^s(X_1),\dots,\tau^s(X_n)), 
\ltens(\Delta) \limp \parr(\Delta',\tau^s(Y_1),\dots,\tau^s(Y_n)),  
\Gamma, \Delta  \\
\vdash
\Gamma', \Delta', \tau^s(Z_1), \dots, \tau^s(Z_n).
\end{array}
$$
This is easily provable, using cut formulas $\tau^s(X_i) \parr \tau^s(Y_i) \limp \tau^s(Z_i).$
\end{proof}

\begin{lemma}
\label{lm:conserv-hollow}
For every hollow FILL sequent $X$, the formula $\punit \limp \tau^s(X)$ 
is provable in FILL.
\end{lemma}

\begin{lemma}
\label{lm:conserv-id}
Let $X[~]$ be a hollow positive FILL context, let $\Sc$ be a multiset
of hollow FILL sequents. Then each of the following formulas
is provable in FILL: 
$$
\tau^s(X[A \seq \Sc, A])
\qquad
\tau^s(X[\punit \seq \sunit])
\qquad
\tau^s(X[\sunit \seq \tunit])
$$
\end{lemma}
\begin{proof}
By induction on $X[~]$, and utilising Lemma~\ref{lm:conserv-hollow}. 
\end{proof}

\begin{lemma}
\label{lm:conserv-non-branching}
Suppose $F \limp G$ is provable in FILL. 
Then for every positive FILL context $X[~]$,
the formula $\tau^s(X[F]) \limp \tau^s(X[G])$ is provable in FILL.
\end{lemma}

\begin{lemma}
\label{lm:conserv-branching}
Suppose $F \ltens G \limp H$ is provable in FILL.
Then for every positive FILL context $X_1[~], X_2[~], X[~]$ such
that $X[~] \in X_1[~] \merge X_2[~]$, we have that
$$
\tau^s(X_1[F]) \ltens \tau^s(X_2[G]) \limp \tau^s(X[H])
$$
is also provable in FILL.
\end{lemma}
\begin{proof}
By induction on the structure of $X[~].$
Suppose $X[~]$
is
$$
\Gamma, \Delta \seq \Sc, X'[~]
$$
and $X_1[~]$ and $X_2[~]$ are, respectively, 
$$
(\Gamma \seq \Sc_1, X_1'[~])
\qquad
(\Delta \seq \Sc_2, X_2'[~])
$$
where $X'[~] \in X_1'[~] \merge X_2'[~]$ and
$\Sc \in \Sc_1 \merge \Sc_2.$
By the induction hypothesis, we have 
$$
\tau^s(X_1'[F]) \ltens \tau^s(X_2'[G]) \limp \tau^s(X'[H]).
$$
To prove $\tau^s(X_1[F]) \ltens \tau^s(X_2[G]) \limp \tau^s(X[H])$ it is enough
to prove the following sequent:
\begin{equation}
\label{eq:seq1}
\begin{array}{l}
\ltens(\Gamma) \limp \parr(\tau^s(\Sc_1), \tau^s(X_1'[F])), 
\ltens(\Delta) \limp \parr(\tau^s(\Sc_2), \tau^s(X_2'[G])), 
\Gamma, \Delta \\
\vdash 
\tau^s(\Sc), \tau^s(X'[H])
\end{array}
\end{equation}
By Lemma~\ref{lm:conserv-merge}, we can show that
$\tau^s(\Sc_1) \parr \tau^s(\Sc_2) \limp \tau^s(\Sc).$ Therefore, to prove 
sequent (\ref{eq:seq1}) it is enough to prove the following sequent:
$$
\begin{array}{l}
\ltens(\Gamma) \limp \parr(\tau^s(\Sc_1), \tau^s(X_1'[F])), 
\ltens(\Delta) \limp \parr(\tau^s(\Sc_2), \tau^s(X_2'[G])), 
\Gamma, \Delta \\
\vdash 
\tau^s(\Sc_1), \tau^s(\Sc_2), \tau^s(X'[H])
\end{array}
$$
which is provable by straightforward applications of logical rules
and the cut rule with the cut formula 
$\tau^s(X_1'[F]) \ltens \tau^s(X_2'[G]) \limp \tau^s(X'[H])$
(which is provable by the induction hypothesis). 
\end{proof}

Now we are ready to prove the statement of Theorem~\ref{thm:filldn-conserv}: 
For every FILL formula $B$, $B$ is provable in FILL if and only if
it is provable in $\filldn.$
\begin{proof}
One direction, from FILL to $\filldn$, follows from the fact that any FILL derivation
is also a derivation in $\biillsn$, and Theorem~\ref{thm:separation}.
For the reverse direction, we show that every derivation of $X$ in $\filldn$ corresponds to 
a derivation of $\vdash \tau^s(X)$ in FILL; hence every valid formula in $\filldn$ is also
valid in FILL. We do this by induction on the height of derivations in $\filldn.$
The base cases where the derivation ends with $id^d$, $\punit^d_l$ or $\tunit^d_r$ follow
from Lemma~\ref{lm:conserv-id}. 
For the inductive cases, we first show that every rule in $\biilldn$ corresponds to
a valid sequent in FILL. 
For a non-branching rule, with premise $U$ and conclusion $V$, the corresponding
implication is $\tau^s(U) \vdash \tau^s(V)$. For a branching rule, with premises
$U$ and $V$, and conclusion $W$, the corresponding implication is 
$\tau^s(U), \tau^s(V) \vdash \tau^s(W).$
Thus given a derivation of $X$ ending with a branching rule:
$$
\infer[\rho]
{X}
{\deduce{X_1}{\vdots} & \deduce{X_2}{\vdots}}
$$
the translation takes the form:
$$
\infer[cut]
{\vdash \tau^s(X)}
{
\deduce
{\vdash \tau^s(X_1)}
{(1)}
&
\infer[cut]
{\tau^s(X_1) \vdash \tau^s(x)}
{
 \deduce{\vdash \tau^s(X_2)}{(2)}
 &
 \deduce{\tau^s(X_1), \tau^s(X_2) \vdash \tau^s(X)}{(3)}
}
}
$$
Sequents (1) and (2) are provable by the induction hypothesis, so it is enough
to show we can always prove sequent (3). We show here a case where the derivation
in $\filldn$ ends with $\limp_l$; the other cases are similar. 
So suppose the derivation in $\filldn$ ends with: 
$$
\infer[\limp_l^d]
{X[\Gamma, \Delta, A \limp B \seq \Tc]}
{
X_1[\Gamma \seq A, \Tc_1] & X_2[\Delta, B \seq \Tc_2]
}
$$
where $\Tc, \Tc_1$ and $\Tc_2$ are multisets of intuitionist formulas/sequents,
$X[~] \in X_1[~] \merge X_2[~]$, and $\Tc \in \Tc_1 \merge \Tc_2.$
By Lemma~\ref{lm:conserv-branching}, 
it is enough to show that the following formula is valid in FILL:
$$
(\ltens(\Gamma) \limp A \parr \tau^s(\Tc_1))
\ltens
(\ltens(\Delta) \ltens B \limp \tau^s(\Tc_2))
\limp
(\ltens(\Gamma, \Delta, A \limp B) \limp \tau^s(\Tc)).
$$
This in turn reduces to proving the sequent:
$$
\ltens(\Gamma) \limp A \parr \tau^s(\Tc_1), 
\ltens(\Delta) \ltens B \limp \tau^s(\Tc_2),
\Gamma, \Delta, A \limp B 
\vdash
\tau^s(\Tc).
$$
This sequent can be proved (in a bottom-up fashion) by 
a cut with the provable formula (by Lemma~\ref{lm:conserv-merge})
$\tau^s(\Tc_1) \parr \tau^s(\Tc_2) \limp \tau^s(\Tc)$, followed
by straightforward applications of introduction rules.
$$
\infer[\mbox{Lem. \ref{lm:conserv-merge}}]
{\ltens(\Gamma) \limp A \parr \tau^s(\Tc_1), 
\ltens(\Delta) \ltens B \limp \tau^s(\Tc_2),
\Gamma, \Delta, A \limp B 
\vdash
\tau^s(\Tc)}
{
\infer[\limp_l]
{
\ltens(\Gamma) \limp A \parr \tau^s(\Tc_1), 
\ltens(\Delta) \ltens B \limp \tau^s(\Tc_2),
\Gamma, \Delta, A \limp B 
\vdash
\tau^s(\Tc_1), \tau^s(\Tc_2)
}
{
\deduce{\ltens(\Gamma) \limp A \parr \tau^s(\Tc_1), 
\Gamma \vdash  A, \tau^s(\Tc_1)}{(4)}
&
\deduce{
\ltens(\Delta) \ltens B \limp \tau^s(\Tc_2),
\Delta, B 
\vdash
\tau^s(\Tc_2)
}{(5)}
}
}
$$
The derivations for sequents (4) and (5) are easy and omitted here. 
\end{proof}


\end{document}